\documentclass[11pt]{article}
\usepackage[utf8]{inputenc}
\usepackage[english]{babel} 
\usepackage[subtle, mathspacing=normal]{savetrees} 

\usepackage{amsfonts,latexsym,amsthm,amssymb,amsmath,amscd,euscript,mathrsfs,graphicx}
\usepackage[normalem]{ulem} 
\usepackage{framed}
\usepackage{fullpage}
\usepackage{color}
\usepackage{amsmath}
\usepackage{enumitem}
\usepackage{url}
\usepackage{doi}
\usepackage{todonotes}
\usepackage{xspace}
\usepackage{times}
\usepackage{cite}

\usepackage{subcaption}
\usepackage{mwe}

\newcommand{\defn}[1]{\emph{\textbf{{\boldmath #1}}}\xspace}

\renewcommand{\paragraph}[1]{\vspace{0.1in}\noindent{\bf \boldmath #1}} 

\newtheorem{theorem}{Theorem}
\newtheorem{proposition}{Proposition}
\newtheorem{lemma}{Lemma}

\newtheorem{corollary}{Corollary}

\theoremstyle{definition}

\theoremstyle{remark}
\newtheorem*{remark}{Remark}

\newcommand{\E}{\mathbb{E}}

\renewcommand{\log}{\lg}

\newcommand{\poisson}{\text{Pois}}

\newcommand{\poly}{\operatorname{poly}}
\newcommand{\polylog}{\operatorname{polylog}}

\renewcommand{\epsilon}{\varepsilon}

\title{Linear Probing Revisited:\\ Tombstones Mark the Death of Primary Clustering\\
\vspace{1.2ex}}

\author{Michael A. Bender\footnote{Supported in part by NSF grants 
CCF-2106827, 
CCF-1725543, 
CSR-1763680,  
CCF-1716252, and 
CNS-1938709.   
}\\
{\small Stony Brook}
\and
Bradley C. Kuszmaul\\
{\small Google Inc.}
\and 
William Kuszmaul\footnote{Supported in part by an NSF GRFP fellowship and a Fannie and John Hertz Fellowship. Research was partially sponsored by the United States Air Force Research Laboratory and was accomplished under Cooperative Agreement Number FA8750-19-2-1000. The views and conclusions contained in this document are those of the authors and should not be interpreted as representing the official policies, either expressed or implied, of the United States Air Force or the U.S. Government. The U.S. Government is authorized to reproduce and distribute reprints for Government purposes notwithstanding any copyright notation herein.}\\
{\small MIT}
}
\date{}
\begin{document}
\maketitle

\smallskip

\begin{abstract}
First introduced in 1954, the linear-probing hash table is among the oldest data structures in computer science, and thanks to its unrivaled data locality, linear probing continues to be one of the fastest hash tables in practice.  It is widely believed and taught, however, that linear probing should never be used at high load factors; this is because of an effect known as primary clustering which causes insertions at a load factor of $1 - 1 /x$ to take expected time $\Theta(x^2)$ (rather than the intuitive running time of $\Theta(x)$).  The dangers of primary clustering, first discovered by Knuth in 1963, have now been taught to generations of computer scientists, and have influenced the design of some of the most widely used hash tables in production.

We show that primary clustering is not the foregone conclusion that it is reputed to be. We demonstrate that seemingly small design decisions in how deletions are implemented have dramatic effects on the asymptotic performance of insertions: if these design decisions are made correctly, then even if a hash table operates continuously at a load factor of $1 - \Theta(1/x)$, the expected amortized cost per insertion/deletion is $\tilde{O}(x)$. This is because the tombstones left behind by deletions can actually cause an \emph{anti-clustering} effect that combats primary clustering. Interestingly, these design decisions, despite their remarkable effects, have historically been viewed as simply implementation-level engineering choices.

We also present a new variant of linear probing (which we call graveyard hashing) that completely eliminates primary clustering on \emph{any} sequence of operations: if, when an operation is performed, the current load factor is $1 - 1/x$ for some $x$, then the expected cost of the operation is $O(x)$. Thus we can achieve the data locality of traditional linear probing without any of the disadvantages of primary clustering. One corollary is that, in the external-memory model with a data blocks of size $B$, graveyard hashing offers the following remarkably strong guarantee: at any load factor $1 - 1/x$ satisfying $x = o(B)$, graveyard hashing achieves $1 + o(1)$ expected block transfers per operation. In contrast, past external-memory hash tables have only been able to offer a $1 + o(1)$ guarantee when the block size $B$ is at least $\Omega(x^2)$.

Our results come with actionable lessons for both theoreticians and practitioners, in particular, that well-designed use of tombstones can completely change the asymptotic landscape of how the linear probing behaves (and even in workloads without deletions). 
\end{abstract}

\thispagestyle{empty} 
\newpage
\pagenumbering{arabic}

\section{Introduction}

The linear probing hash table~\cite{Sedgewick83,Sedgewick90, Kruse84, LewisDe91, Weiss00, MainSa01, Standish95, DrozdekSi95, CormenLeRi09, McMillan14, GoodrichTa15, TremblaySo84} is among the most fundamental data structures to computer science. The hash table takes the form of an array of some size $n$, where each \defn{slot} of the array either contains an element or is empty (i.e., is a \defn{free slot}).
To insert a new element $u$, the data structure computes a hash $h(u) \in [n]$, and places $ u $ into the first free slot out of the sequence $h(u), \; h(u)+1, \; h(u)+2, \; \ldots \;$ (modulo $n$).
Likewise, a query for $u$ simply scans through the slots, terminating when it finds either $u$ or a free slot. 

There are two ways to implement deletions: immediate compactions and  tombstones.
An immediate compaction rearranges the neighborhood of elements around a deletion to make it as though the element was never there~\cite{Knuth98Vol3,JimenezMa18}. 
Tombstones, the approach we default to in this paper, replaces the deleted element with a \defn{tombstone}~\cite{Knuth98Vol3,Peterson57}.   Tombstones interact asymmetrically with queries and insertions: queries treat a tombstone as being a value that does not match the query, whereas insertions treat the tombstone as a free slot. 
In order to prevent slowdown from an accumulation of tombstones, the table is occasionally rebuilt in order to clear them out. 

\paragraph{The appeal of linear probing.}
Linear probing was discovered in 1954 by IBM researchers Gene Amdahl, Elaine McGraw, and Arthur Samuel, who developed the hash table design while writing an assembly program for the IBM 701 (see discussion in \cite{Knuth63,Knuth98Vol3}).
The discovery came just a year after fellow IBM researcher Hans Peter Luhn introduced chained hashing (also while working on
the IBM 701) \cite{Knuth98Vol3}.\footnote{Interestingly, Luhn's discussion of chaining may be the first known use of linked lists \cite{Knuth98Vol3, Knuth97Vol1}.
Linked lists are often misattributed~\cite{WikipediaLinkedList} as having been invented by Newell, Shaw, and Simon~\cite{NewellSi56, NewellSh57} during the RAND Corporation's development of the IPL-2 programming language in 1956 (see discussion in~\cite{Knuth97Vol1}).} Linear probing shared the relative simplicity of chaining, while offering improved space utilization by avoiding the use of pointers.

The key property that makes linear probing appealing, however, is its data locality: each operation only needs to access one localized region of memory.
In 1954, this meant that queries could often be completed at the cost of accessing only a single drum track \cite{Knuth63, Peterson57}. 
In modern systems, it means that queries can often be completed in just a single cache miss~\cite{RichterAlDi15}.\footnote{Moreover, even when multiple cache misses occur, since those misses are on adjacent cache lines, hardware prefetching can mitigate the cost of the subsequent misses. Both the Masstree~\cite{MaoKoMo12} and the Abseil B-tree \cite{Abseil17} treat the effective cache-line size as 256 bytes even though the hardware cache lines are 64 bytes.}
The result is that, over the decades, even as computer architectures have changed and as the study of hash tables has evolved into one of the richest and most studied areas of algorithms (see related work in Section~\ref{sec:related}), linear probing has persisted as one of the best performing hash tables in practice~\cite{RichterAlDi15}.

\paragraph{The drawback: primary clustering.}
Unfortunately, the data locality of linear probing comes with a major drawback known as \defn{primary clustering} \cite{Knuth63, KonheimWe66}. Consider the setting in which one fills a linear-probing hash table up to a \defn{load factor} of $1 - 1/x$ (i.e., there are $(1 - 1/x)n$ elements)
and then performs one more insertion. Intuitively, one might expect the insertion to take time $\Theta (x)$, since one in every $x$ slots are free. As Knuth \cite{Knuth63} discovered in 1963\footnote{Knuth's result \cite{Knuth63} was never formally published, and was subsequently rediscovered by Konheim and Weiss in 1966~\cite{KonheimWe66}.}, however, the insertion actually runs much slower, taking expected time $\Theta(x^2)$.

The reason for these slow insertions is that elements in the hash table have a tendency to cluster together into long runs; this is known as primary clustering. 
Primary clustering is often described as a ``winner keeps winning'' phenomenon, in which the longer a run gets,
the more likely it is to accrue additional elements; see e.g.,~\cite{CormenLeRi09, DrozdekSi95, Kruse84, LewisDe91, Morris68, TremblaySo84, WikipediaPrimaryClustering, Smith04}.
As Knuth discusses in \cite{Knuth98Vol3}, however, winner-keeps-winning is not the main cause of
primary clustering.\footnote{This can be easily seen by the following thought experiment. Consider the length $\ell$ of the run at a given position in the hash table,
and then consider a sequence of $(1 - 1/x)n$ insertions where we model each insertion as having probability $(1 + \ell)/n$ of increasing the length of the run by $1$ (here
$\ell$ is the length of the run at the time of the insertion). The expected length $\ell$ of the run at the end of the insertions won't be $\Theta(x^2)$. In fact,
rather than being an asymptotic function of $x$, it will simply be $\Theta(1)$.}
The true culprit is the globbing together of runs: a single insertion can connect together two already long runs into a new run that is substantially longer.

Interestingly, primary clustering is an asymmetric phenomenon, affecting insertions but not queries. Knuth showed that if a query is performed on
a random element in the table, then the expected time taken is $\Theta(x)$~\cite{Knuth63}. This can be made true for all queries (including negative queries)
by implementing a simple additional optimization: rather than placing elements at the end of a run on each insertion, order the elements within each run by their hashes. This technique, known as \defn{ordered linear probing}~\cite{AmbleKn74,CelisLaMu85}, allows a query to terminate as soon as it reaches an element whose hash is sufficiently large.\footnote{Contemporary works sometimes also refer to this as \defn{Robin hood hashing}, but as noted in \cite{CelisLaMu85}, Robin hood hashing is actually a generalization of ordered linear probing
to other open-addressing schemes such as double hashing, uniform probing, etc.}
Whereas insertions must traverse the entire run, queries need not.

Knuth's result predates the widespread use of asymptotic notation in algorithm analysis.
Thus Knuth derives not only the asymptotic formula of $\Theta(x^2)$ for expected insertion time, 
but also a formula of 
\begin{equation}
\frac{1}{2}\left(1 + \frac{1}{x^2}\right),
\label{eq:knuthexact}
\end{equation}
which is exact up to low order terms. 
There has also been a great deal of follow-up work determining detailed tail bounds and other generalizations of Knuth's result; see e.g.,~\cite{FlajoletPoVi98,Viola05,KonheimWe66,MendelsonYe80,JansonVi16, ViolaPo96, Viola10}.

\paragraph{The practical and cultural effects of primary clustering.}
Primary clustering is taught extensively in both theoretical and practical courses~\cite{Fisher01, Kesden07, DemaineLe05, Demaine12-6897-lecture10, GriesJa14, Bauer15, deGreef17, Erickson17, Mount19, Mitra21, Sullivan21, Schwarz21}. Many textbooks not only teach primary clustering~\cite{Sedgewick83,Sedgewick90, Kruse84, LewisDe91, Weiss00, MainSa01, Standish95, DrozdekSi95, CormenLeRi09, McMillan14, GoodrichTa15, TremblaySo84}, but also teach the full formula \eqref{eq:knuthexact} for insertion time\cite{Sedgewick83,Sedgewick90, Kruse84, LewisDe91, Weiss00, MainSa01, Standish95, DrozdekSi95, TremblaySo84}.\footnote{Some books also go through worked examples or tables of \eqref{eq:knuthexact} to give intuition for how badly linear probing scales, e.g.,~\cite{Sedgewick90, LewisDe91, Standish95, TremblaySo84}.} (For a more detailed summary of how courses and textbooks teach linear probing and primary clustering, see Appendix \ref{sec:books}.)


Because \eqref{eq:knuthexact} is exact, it is viewed as representing the full picture of how linear probing behaves at high load factors. 
One consequence is that there has been little empirical work on analyzing the asymptotics of real-world linear probing at high load factors.
(And for good reason, what would be the point of verifying what is already known?)
As Sedgewick observed in 1990~\cite{Sedgewick90}, it is not even clear that the classic formula has 
been empirically verified at high load factors.

A common recommendation~\cite{Weiss00, DrozdekSi95, CormenLeRi09, WikipediaLinearProbing, Fisher01, Kesden07, GriesJa14, Bauer15, deGreef17, Sullivan21} is that, in order to avoid primary clustering,
one should use \defn{quadratic probing} \cite{Maurer68, HopgoodDa72, WikipediaQuadraticProbing} instead.
Whereas linear probing places each element $ x $ in the first available position out of the sequence 
$ h (x),h (x) +1,h (x) +2,h(x) + 3,\ldots$, quadratic probing uses the first available position out of a quadratic sequence  such as
$h(x), h(x) + 1, h(x) + 4, h(x) + 9, \ldots, h(x) + k^2, \ldots$ or $h(x), h(x)+1, h(x)+3, \ldots, h(x)+k(k-1)/2, \ldots$.
By using a more spread out sequence of probes, quadratic probing seems to eliminate primary clustering
in practice. In doing so, quadratic probing also compromises the most attractive trait of linear probing, its data locality. This tradeoff is typically viewed as unfortunate but necessary.

The dangers of primary clustering (and the advice of using quadratic probing as a solution) have been taught to generations of computer scientists over roughly six decades. 
The folklore advice has shaped some of the most
widely used hash tables in production, including the high-performance hash tables authored by both Google~\cite{Abseil17} and Facebook~\cite{BronsonSh19}.\footnote{Both hash tables are implemented using variants of quadratic probing.}
The consequence is that primary clustering---along with the design compromises made to avoid it---has a first-order impact on the performance of hash tables used by millions of users every day. 

\paragraph{This paper: why primary clustering isn't a foregone conclusion.}
In this paper, we reveal that primary clustering is not the fixed and universal phenomenon that it is reputed to be.
When implementing linear probing, there is a small set of design decisions that are typically treated as implementation-level engineering choices. We show that these decisions actually have a remarkable 
 effect on performance: even if a workload operates continuously at a load factor of
$1 - \Theta(1/x)$, if the design decisions are made correctly, then the expected amortized cost per insertion can be decreased all the way to $\tilde{O}(x)$.
Our results come with actionable lessons for practitioners, indicating that implementation-level decisions can have unintuitive asymptotic
consequences on performance.\footnote{The right thing to do may even be the \emph{opposite} of what the literature 
recommends \cite{JimenezMa18,AttractiveChaos19}.} 

We also present a new variant of linear probing, which we call \defn{graveyard hashing}, that completely eliminates primary clustering on any sequence of operations: if, when an operation is performed, the current load factor is $1 - 1/x$ for some $x$, then the expected cost of the operation is $O(x)$. 

Thus we achieve the data locality of traditional linear probing without any of the disadvantages of primary clustering. One corollary is that, in the external-memory model with a data blocks of size $B$, graveyard hashing offers the following remarkably strong guarantee: at any load factor $1 - 1/x$ satisfying $x = o(B)$, graveyard hashing achieves $1 + o(1)$ expected block transfers per operation. In contrast, past external-memory hash tables have only been able to offer a $1 + o(1)$ guarantee when the block size $B$ is at least $\Omega(x^2)$ \cite{JensenPa08}.

\paragraph{What the classical analysis misses.}
Classically, the analysis of linear probing considers the costs of insertions in an insertion-only workload.
Of course, the fact that the final insertion takes expected time $\Theta(x^2)$ doesn't mean that all of
the insertions do; most of the insertions are performed at much lower load factors, and the average cost is only $\Theta(x)$. 

The more pressing concern is what happens for workloads that operate continuously at high load factors, 
for example, the workload in which a user first fills the table to a load factor of $ 1-1/x $,
and then alternates between insertions and deletions indefinitely. Now almost all of the insertions
are performed at a high load factor. Conventional wisdom has it that these insertions must therefore
all incur the wrath of primary clustering.


This conventional wisdom misses an important point, however, which is that the tombstones created by deletions actually substantially change the combinatorial structure of the hash table. Whereas insertions add elements at the ends of runs, deletions tend to place tombstones in the middles of runs.
If implemented correctly, then the anti-clustering effects of deletions actually outpace the clustering effects of insertions. 

We call this new phenomenon \defn{primary anti-clustering}.
The effect is so powerful that, as we shall see, it is even worthwhile to simulate deletions in insertion-only workloads 
by prophylactically adding tombstones. 

Our results flip the narrative surrounding deletions in hash tables: whereas past work on analyzing tombstones  \cite{JimenezMa18,AttractiveChaos19} has focused on showing that tombstones do not \emph{degrade} performance in various open-addressing-based hash tables, we argue that tombstones 
actually \emph{help} performance. By harnessing the power of tombstones in the right way, we can rewrite the asymptotic landscape of linear probing. 

\subsection{Results}
We begin by giving nearly tight bounds on the amortized performance of ordered linear probing. 

There are two design decisions that affect performance: (1) the use of tombstones (which we assume by default) and (2) the frequency with which the hash table is rebuilt and the tombstones are cleared out. 
Thus, in addition to the size parameter $n$ and the load-factor parameter $x$, our analysis uses a \defn{rebuild window size} parameter $R$, which is the number of insertions that occur between rebuilds. 

The parameter $R$ is classically set to be $n / (2x)$, since this means that the number of tombstone cannot affect the asymptotic load factor.
Each rebuild can be implemented in time $\Theta(n)$, so the rebuilds only contribute amortized $\Theta(x)$ time per insertion, which is a low-order term. 


\paragraph{\boldmath A subquadratic analysis of linear-probing insertions.}
Our first result considers the classical setting $R = n/(2x)$ and analyzes a \defn{hovering workload}, i.e., an alternating sequence of
inserts/deletes at a load factor of $1 - 1/x$. 

We prove that the expected amortized cost of each insertion is $\tilde{O}(x^{1.5})$.
This is tight, which we establish with a lower bound of $\Omega (x ^ { 1.5 }\sqrt {\log\log x }) $. A surprising takeaway is that the answer is both $\tilde{O}(x^{1.5})$ and $\omega(x^{1.5})$.

This first result is already substantially faster than the classical $\Theta(x^2)$ bound, but it still has several weaknesses. The first (and most obvious) weakness is that we are still not achieving the ideal bound of $\tilde{O}(x)$.  
The second weakness is that, although the result applies to hovering workloads, it doesn't generalize to \emph{arbitrary} workloads, as can be seen with the following
pathological example: consider a workload in which every rebuild window consists of $R - 1$ insertions followed by $R$ deletions followed by $1$ insertion. The first $R - 1$ insertions in each rebuild window cannot benefit at all from tombstones and thus necessarily incur $\Theta(x^2)$ expected time each.

It turns out that both of these weaknesses can be removed if we simply use a larger rebuild window size $R $.
Intuitively, the larger the $R$, the more time there is for
tombstones to accumulate and the better the insertions perform. On the other hand, tombstone accumulation is \emph{precisely} the reason that $R$ is classically set to be small, since it breaks the classical analysis and potentially tanks the performance of queries.

We show that the sweet spot is to set $R = n/\polylog(x)$. Here, the expected amortized cost per insertion drops all the way to $\tilde{O}(x)$, while queries continue to take expected time $O(x)$. Once again, and somewhat surprisingly, the low-order factors are an artifact of reality rather than merely the analysis: no matter the value of $R$ used, either the average insertion cost or the average query cost must be $\omega(x)$.

The bound of $\tilde{\Theta}(x)$ holds not only for hovering workloads, but also for \emph{any workload} that maintains a load factor of at most $1 - 1/x$. Note that here we are analyzing a table in which the capacity $n$ is fixed, and the load factor is permitted to vary over time. 
It's interesting to see about how the pathological case described above is avoided here. Because the rebuild window is so large, the only way there can be 
a long series of insertions without deletions is if most of them are performed at low load factors, meaning that they are not slow after all.

\paragraph{A surprising lesson: linear probing is already faster than we thought.}
The core lesson of our results is that linear probing is far less affected by primary clustering than the classical analysis would seem to suggest. 
Although the classic $\Theta(x^2)$ bound is mathematically correct, it does not accurately represent the amortized cost of insertions at high load factors. This suggests that  conclusions that are taught in courses and textbooks, namely that linear probing scales poorly to high load factors, and that alternatives with less data locality such as quadratic probing should be used in its place, stem in part from an incomplete understanding of linear probing and warrant revisiting. 

The second lesson is that small implementation decisions
can substantially change performance. From a software engineering perspective,
our results suggest two simple optimizations (the use of tombstones and the use of large rebuild windows)
 that should be considered in any implementation of linear probing. 

Interestingly, tombstones (and even  relatively large rebuild windows) are already present in some hash tables. Thus, one interpretation of primary anti-clustering is as a phenomenon that, to some degree, already occurs around us, 
but that until now has gone apparently unnoticed.

\paragraph{Graveyard hashing: a variant of linear probing with ideal performance.}
Our final result is a new version of linear probing, which we call \defn{graveyard hashing}, that fully eliminates primary clustering on any sequence of operations. The key insight is that, by artificially inserting
extra tombstones (that are not created by deletions), we can ensure that every insertion has good
expected behavior. 

Insertions, deletions, and queries are performed in exactly the same way as for
standard ordered linear probing. The difference is in how we implement rebuilds.
In addition to cleaning out tombstones, rebuilds are now also responsible for 
inserting $\Theta(n / x)$ \emph{new tombstones} evenly spread across the key space. 

Graveyard hashing adapts dynamically to the current load factor of the table,
performing a rebuild every time that $x$ changes by a constant factor.
The running time of each operation is a function of whatever the load factor
is at the time of the operation. If a query/insert/delete occurs at a load factor of $1 - 1/x$, then it takes expected time $O(x)$ (even in insertion-only workloads). 
Graveyard hashing can also be implemented to resize dynamically, so that it is always at some target load factor of $1 - \Theta(1 / x)$.

\paragraph{Coming full circle: improved external-memory hashing.}
As we mentioned at the outset, one of the big advantages of linear probing is its data locality (e.g., good cache or I/O performance). 

Data locality is formalized via the external-memory model,  first introduced by Aggarwal and Vitter in 1988~\cite{AggarwalVi88}: 
a two-level memory hierarchy is comprised of a small, fast \defn{internal memory} 
and a large, slow \defn{external memory}; 
blocks can only be read and modified when they are in internal memory, 
and the act of copying a block from external memory into internal memory is referred to as a \defn{block transfer}.
The model has two parameters, the number $B$ of
of records that fit into each block and the number $M$ of records that fit in internal memory. 
Performance is measured by the number of block transfers that a given algorithm 
or data structure incurs. 
This model can be used to capture an algorithm's I/O performance (internal memory is RAM, external memory is disk, and block transfers are I/Os) or cache performance (internal memory is cache, external memory is RAM, and block transfers are cache misses).

Ideally, a hash table incurs only
amortized expected $1 + o(1)$ block transfers per operation, even when  
supporting 
a  high load factor $1 - 1/x$.\footnote{Many hash tables that are otherwise very appealing perform poorly on this front. For example, if one uses cuckoo hashing, then negative queries require $\geq 2$ block transfers, and insertions at high load factor require $\omega(1)$ block transfers~\cite{FotakisPaSa05,DietzfelbingerWe07}.} 
We call the problem of achieving these guarantees the \defn{space-efficient external-memory hashing problem}.
Standard linear probing is a solution when either $x$ is a (small) constant, or the block size $B$ is very large ($B = \omega(x^2)$), but otherwise, due to primary clustering, it is not \cite{PaghWeYi14}.

For $B \neq \omega(x^2)$, the  state of the art for space-efficient external-memory hashing is due to Jensen and Pagh \cite{JensenPa08}. 
They give an elegant construction showing that, if the block size $B$ is $\Theta(x^2)$, then it is possible to achieve
amortized expected $1 + O(1 / x)$ block transfers per operation, while maintaining a load factor of $1 - O(1/x) $. (In contrast, standard linear probing requires $B = \Omega( x^3)$ to achieve the same $1 + O(1/x)$ result.) However, if $B = o(x^2)$, then no solutions to the problem are known.

Graveyard hashing enables linear probing to be used directly as a solution to the space-efficient external-memory hashing problem, matching Jensen and Pagh's bound for $B = \Theta(x^2)$, and  offering an analogous guarantee for arbitrary block sizes $B > \omega(x)$.
If  $ B =\Theta (x k)$ for some $k > 1$, then the amortized expected 
cost of each operation is $1 + O(1 / k)$ block transfers. This means that, even if the block size $B$ 
is only \emph{slightly} larger than the load factor parameter $ x $, we still get
$1 + o(1)$ block transfers per operation.

Additionally, graveyard hashing is cache oblivious \cite{FrigoLePr99, PaghWeYi14}, meaning that the 
block size $B$ need not be known by the data structure. Consequently, if a system has
a multi-level hierarchy of caches, each of which may have a different set of parameters $B$ and $ M$, then
the guarantee above applies to every level of cache hierarchy.



\section{Notation and Conventions}

We say that an event occurs with probability $1 - 1 / \poly(j)$ for some parameter $j$ if, for any positive constant $c$, the event occurs with probability $1 - O(1 / j^c)$.\footnote{The
constants used to define the event may depend on $c$.}
Throughout the paper, we use standard
interval notation, where $[m]$ means $\{1, 2, \ldots, m\}$, $[i, j]$ means $\{i, i + 1, \ldots, j\}$, $(i, j]$ means $\{i + 1, i + 2, \ldots, j\}$, etc.

When discussing an ordered linear probing hash table, we use $n$
to denote the number of \defn{slots} (i.e., \defn{positions}). We use
$1 - 1/x$ to refer to the \defn{load factor} (i.e., the fraction of slots that are 
taken by elements), and $R$ to refer to the \defn{rebuild-window size} (i.e., the number of insertions that must occur before a rebuild is performed). 
We use $S$ to denote the sequence of operations being performed, and we refer $S$ as the \defn{workload}.

The operations on the hash table make use of a \defn{hash function} $h:U \rightarrow [n]$, where $U$ is the universe of possible \defn{keys} (also known as \defn{records} or \defn{elements}). We shall assume that $h$ is uniform and fully independent, but as we discuss in later sections, our results also hold for natural families of hash functions such as tabulation hashing (and, for the analysis of graveyard hash tables, also 5-independent hashing). 
We can also refer to the \defn{hash} $h(u)$ of either a tombstone $u$ (i.e., the hash of the element whose deletion created $u$) or of an operation $u$ (i.e., the hash of the element on which the operation takes place).

Each slot in the hash table can either contain a key, be empty (i.e., a free slot), or contain a tombstone. 
Any maximal contiguous sequence of non-empty slots forms a \defn{run}.
With ordered linear probing, the keys/tombstones in each run are always stored in order of their hash.

Our analysis will often discuss \defn{sub-intervals} $I = [i, j] \subseteq [n]$ of the slots in the hash table. We say that an element $u$ \defn{hashes to} $I$ if $ h (u)\in I $ (but this does not necessarily mean that $ u $ resides in one of the slots $ I $). We say that an interval $ I $ is \defn{saturated} if it is a subset of a run.

Formally, operations in an ordered linear probing hash table are implemented as follows. A query for a key $u$ examines positions $h(u), h(u) + 1, \ldots$ until it either finds $u$ (in which case the query returns true), finds an element with hash greater than $h(u)$ (in which case the query returns false), or finds a free slot (in which case the query also returns false). 
A deletion of a key $u$ simply performs a query to find the key, and then replaces it with a tombstone. Finally, an insertion of a key $u$ examines positions $h(u), h(u) + 1, \ldots$ until it finds the position $j$ where $u$ belongs in the run; it then inserts $u$ into that position, and shifts the elements in positions $j, j + 1, j + 2, \ldots$ each to the right by one until finding either a tombstone or a free slot.
We say that the insertion \defn{makes use} of that tombstone/free slot. Finally, rebuilds are performed every $R$ insertions, and a rebuild simply restructures the table to remove all tombstones. Throughout the paper, $n$ is fixed and does not get changed during rebuilds, although when we describe graveyard hashing (Section~\ref{sec:graveyard}), we also give a version that dynamically resizes $n$. 

There are two ways to handle overflow off the end of the hash table.
One option is to wrap around, meaning that we treat $1$ as being the position that comes after $n$; the other is to extend the table by $o(n)$ (i.e., it ends in slot $n + o(n)$), so that operations never fall off the end of the table. Both solutions are compatible with all of our results. For concreteness, we  assume that the wrap-around solution is used, but to simplify discussion, we treat the slots that we are analyzing as being sufficiently towards the middle of the table that we can use the $<$-operator to compare slots.




\section{Technical Overview}

In this section, we give a technical overview of our analysis of ordered linear probing; we defer our analysis graveyard hashing to Section~\ref{sec:graveyard} in the main body of the paper. We begin by describing the intuition behind Knuth's classic $\Theta(x^2)$ bound. We then turn our attention to sequences of operations that contain deletions, and show that the tombstones left behind by those deletions have a \emph{primary-anti-clustering effect}, that is, they have a tendency to speed up future insertions. One of the interesting components of the analysis is that we  perform a series of problem transformations, taking us from the question of how to analyze ordered linear probing to a seemingly very different question involving the combinatorics of monotone paths on a grid. By applying geometric arguments to the latter, we end up being able to achieve nearly tight bounds for the former.

\subsection{Understanding the classic bounds: a tale of standard deviations}
Suppose we fill an ordered linear-probing hash table from empty up to a load factor of
$ 1-1/x $. Knuth \cite{Knuth63} famously showed that the final insertion in this procedure
takes expected time $\Theta(x^2)$. As discussed in the introduction,
the fact that the insertion takes time $\omega(x)$ can be attributed to primary clustering.

But why does the running time end up being $\Theta(x^2)$ specifically?
This turns out to be a result of how standard deviations work. Consider 
an interval $I$ of $x^2$ slots in the hash table. The expected number
of items that hash into $ I $ is $(1 - 1 / x) \, x^2 = |I| - x$. On the other hand,
the standard deviation for the number of such items is $\Theta(x)$. It follows that,
with probability $\Omega (1) $, the number of items that hash into $I$ is 
$\Omega(x)$ greater than $|I|$. If we then consider the interval $I'$
consisting of the $x^2$ slots that follow $I$, then this interval $I'$
must handle not only the items that hash into it, but also the overflow elements from $ I $.
The result is that, with probability $\Omega(1)$,
the interval $ I' $ is fully saturated and forms a run of length $x^2$.

The above argument stops working if we consider intervals of size
$\omega(x^2)$, because the standard deviation on the number of items that hash into $I$
stops being large enough to overflow the interval. The result is that runs of length $\Theta(x^2)$
are relatively common, but that longer runs are not. This is why the expected running time of the insertion
performed at load factor $1 - 1/x$ is $\Theta(x^2)$.

Now suppose that, after reaching a load factor of $1 - 1/x$, we perform a query in our ordered linear-probing hash table. Unlike an insertion, which takes expected time $\Theta(x^2)$, the query takes expected time $\Theta(x)$. We can again see this by looking at standard deviations.

If the query hashes to some position $j$ and takes time $t$, then 
there must be at least $t$ elements $u$ that have hashes $h(u) \le j$ but that reside in positions $j$ or larger. 
Hence, there is some sub-interval $I = [j_0, j]$ of the hash table (ending in position $ j $) that has overflowed by at least $t - 1$ elements, that is, the number of elements that hash to $I$ is at least $|I| + t - 1$. As before, the interval that matters most ends up being the one of 
size $\Theta(x^2)$, and the amount by which it overflows in expectation is 
proportional to the standard deviation $\Theta(x)$ of the number of items that hash into the interval.

Thus, the running times of both insertions and queries are consequences of the same two facts: that (a)~for any interval $I$ of size $x^2$, there is probability $\Omega(1)$ that the interval overflows by
$\Theta(x)$ elements; and that (b)~a given interval of size $\omega(x^2)$ most likely doesn't overflow at all.
The only difference is that the running time of an insertion is proportional
to the \emph{size} of the interval that overflows (i.e., $\Theta(x^2)$), but the running time of a query is proportional
to the \emph{amount} by which the interval overflows (i.e., $\Theta(x)$).

\subsection{Analyzing primary anti-clustering with small rebuild windows}
In this subsection, we consider a hovering workload, that is a sequence of operations that alternates between insertions and deletions on a table with load factor $1 - 1/x$,
and we set the size of the rebuild window to be $R = n / (2x)$ (i.e., the value that it is classically set to). 
Our task is to consider a sequence of $R = n / (2x)$ insertion/deletion pairs between two rebuilds, and to analyze the amortized running times.

\paragraph{Analyzing displacement instead of running time.}
Define the \defn{peak} $p_u$ of an insertion $u$ to be either the
hash of the tombstone that the insertion uses (if the insertion makes use of a tombstone)
or the position of the free slot that the insertion uses (if the insertion makes use of a free slot).
Define the \defn{displacement} $d_u$ of an insertion to be $d_u = p_u - h(u)$. 

One of the subtleties of how displacement is defined is that, if an insertion $ u $ uses a tombstone $ v $, then the displacement measures
the difference between $h(u)$ and the \emph{hash} $h(v)$, rather than the difference between $h(u)$ and the \emph{position} of $ v $.
This ends up being important for how displacement is used in the analysis\footnote{The reason for this is actually very simple. Whenever a deletion $ v $ is performed, the value of $h(v)$ is fixed (it depends only on the
hash function) but the position of $v$ is not (it depends on the other elements in the table, and will change over time). Thus, it is cleaner to measure the deletion's effect on future insertions in terms of $h(v)$ (the thing that is fixed) rather than $v$'s position (the thing that is not).}, but it also means that the displacement of an insertion can potentially be substantially smaller than the running time. For example, if the insertion hashes to position $7$ and makes use of a tombstone with hash $13$ that resides in position $54$,
then the displacement is only $13 - 7 = 6$ but the running time is proportional to $54 - 7 = 47$.

Although we skip the proof for now (see Lemma~\ref{lem:displacementerror}), it turns out that one can bound the expected difference between displacement and running time by $O(x)$. Thus, even though displacement is not always the same as running time, any bound on average displacement also results in a bound on average running time. 

\paragraph{Relating displacement to crossing number.}
Rather than analyzing the displacement of each \emph{individual} insertion, we bound the average displacement
over all $R$ insertions in the rebuild window by relating the displacements to another set of quantities that we call 
the \defn{crossing numbers} $\{c_j\}_{j \in [n]}$. The crossing number $c_j$ counts the number of insertions
$u$ in the rebuild window that have a hash $h(u) < j$ but that have a peak $p_u \ge j$ (we consider even the insertions that are subsequently deleted). Each insertion $u$ 
increments $d_u$ different crossing numbers $c_j$. Thus
$$\sum_{u} d_u =\sum_{j\in [n] }^{} c_j.$$
Because we are analyzing the insertions in a rebuild window of size $R$, the summation on the left side has $R = \Theta(n / x)$ terms,
while the summation on the right side has $n$ terms. Thus, if we consider
a random insertion $ u $ and a random position $ j $, then
$$\E[d_u] = \Theta(x \E[c_j]).$$
If our goal is to establish that the average insertion takes time $o(x^2)$,
then it suffices instead to show that the average crossing number $c_j$ is $o(x)$.

Notice that the ratio between $\E[d_u]$ and $\E[c_j]$ is a function of the rebuild window size $R$; this will come into play later when we consider larger rebuild windows.

\paragraph{Capturing the dependencies between past insertions/deletions.}
What makes the analysis of a given crossing number $c_j$ interesting is the way in which insertions and deletions interact over time.
If an insertion $u$ has hash $h(u) < j$, and there is a tombstone $v$ with hash $h(v) \in [h(u), j)$, then $u$ can make use of the tombstone
and avoid contributing to $c_j$. But, in order to determine whether a given tombstone $v$ is present during the insertion $u$, 
we must know whether any past insertions have already used $v$. That, in turn, depends on which tombstones were present during past insertions, 
resulting in a chain of dependencies between operations over time.

One of the insights in this paper is that the interactions between insertions and deletions over time can be reinterpreted as
an elegant combinatorial problem about paths on a two-dimensional grid. We  now give the transformation.

\paragraph{The geometry of crossing numbers.}
For the sake of analysis, define $ Z $ to be the state that our table
would be in if we performed only the insertions in the rebuild window and not the deletions.

Since $R = n / (2x)$, the load factor of $ Z $ is at most $1 - 1/x + R / n = 1 - 1 / (2x)$, 
which by the classic analysis of linear probing means that the expected distance from any position
to the next (and previous) free slot is $O(x^2)$. 

Now consider some crossing number $ c_j $. Let $j'$ be the position of the closest free slot to the left of $j$ in $Z$, that is, the 
largest $j' < j$ such that position $j'$ is a free slot in $Z$. 
Then the only insertions/deletions
that we need to consider when analyzing $c_j$ are those that hash into the interval
$I = [j' + 1, j)$. 

We know from our analysis of $Z$ that $\E[|I|] = O(x^2)$.
Although $|I|$ is a random variable with mean $\Theta(x^2)$, to simplify our discussion 
in this section, we shall treat $|I|$ as simply deterministically equaling $x^2$.
We also treat $I$ as containing no free slots (even at the beginning of the time window being considered).
In particular, we know from the classical analysis of linear probing that any interval $I$ of size $x^2$ has probability
$\Omega(1)$ of containing no free slots, so there is no point in trying to make use of potential free slots in $I$ for our analysis.

We can visualize the insertions and deletions that hash into $I$ by plotting them in a
two-dimensional grid, as in Figure~\ref{fig:path1}. The vertical axis represents time flowing up from $1$ to $2R$,
and the horizontal axis represents the hash locations in the interval $I$. We draw a blue dot in position $(i, t)$
if the $t$-th operation is an insertion with hash $i \in I$, 
and we draw a red dot in position $(i, t)$ if the $t$-th operation is a deletion
with hash $i \in I$. (Note that most operations do not hash to $I$ and thus do not result in any dot.)

In order for a given insertion $u$ to be able to make use of a given deletion
$v$'s tombstone, it must be that (a)~$h(u) \le h(v)$;\footnote{Recall that ordered linear probing only ever moves elements to the right over time, meaning that a given insertion $u$ will only use a tombstone if that tombstone has hash at least $h(u)$.} that (b)~$u$ occurs temporally after $v$;
and that (c)~$v$'s tombstone is not used by any other insertion temporally before $u$.
The first two criteria (a) and (b) tell us that for a given insertion (i.e., blue dot)
in the grid, the insertion can only make use of tombstones from deletions (i.e., red dots)
that are below it and to its right. 

Define the set of \defn{monotone paths} through the grid to be the set of paths
that go from the bottom left to the top right of the grid, and that never travel downward or leftward. 
Define the  \defn{blue-red deviation} of such a path
to be the number of blue dots below the path minus the number of red dots below the path (see Figure~\ref{fig:path1} for an example).

What do these monotone paths have to do with the crossing number $ c_j $?
Monotone paths with large blue-red deviations serve as witnesses for 
the crossing number $c_j$ also being large. 
Suppose that there is a monotone path $\gamma$ with blue-red deviation $r > 0$. 
Since we assume that $I$ is initially saturated, each of the insertions (i.e., blue dots) below $\gamma $ must either make use of the tombstone 
for a deletion (i.e., red dot) that is also below $\gamma $
or contribute 1 to $c_j$. Since there are $r$ more blue dots than red dots below $\gamma $,
it follows that $c_j \ge r$. 

In fact, this relationship goes in both directions (although the other direction requires a 
bit more work; see Lemma~\ref{lem:crossingsurplus}). If the crossing number $c_j$ takes some value $r$, then there must also exist 
some monotone path with blue-red deviation at least $r$. 
The result is that, if we wish to prove either upper or lower bounds on $\E[c_j]$,
it suffices to instead prove bounds on the largest blue-red deviation
of any monotone path through the grid.

\paragraph{Formalizing the blue-red deviation problem.}
Let us take a moment to digest the combinatorial problem that we have reached,
since on the face of things it is quite different from the problem that we started at.

The expected number of blue dots (and also of red dots) in our grid is $R \cdot |I| / n = \Theta(x)$. 
If we break the grid into $\sqrt{x}$ rows and $\sqrt{x}$ columns (see Figure~\ref{fig:path2} for an example), then each cell of the broken-down grid
expects to contain $\Theta(1)$ blue and red dots. To simplify our discussion here, think of each 
cell as independently containing a Poisson random variable $\poisson(1)$ number of blue dots
and a Poisson random variable $\poisson(1)$ number of red dots.\footnote{This allows for us to ignore two minor issues in our discussion here: (1)~the fact that a single key can potentially be inserted, deleted, and reinserted, resulting in blue and red dots whose horizontal coordinates are deterministically equal; and (2)~the fact that the numbers of blue/red dots in each cell are actually very slightly negatively correlated.}

Furthermore, rather than considering all monotone paths through the grid, we can restrict ourselves
exclusively to the paths that stay on the row and column lines that we have drawn (for an example, see Figure~\ref{fig:path2}). With high probability in $x$, this restriction changes the
maximum blue-red deviation of any path by at most $\tilde{O}(\sqrt{x})$. 

In summary, we have a $\sqrt{x} \times \sqrt{x}$ grid where each cell of the grid contains a
Poisson random variable $\poisson(1)$ number of blue points (resp.\ red points). 
Whereas our original grid was much taller than it was wide (its height was $2R$ and its 
width was $x^2$), our new grid is a $\sqrt{x} \times \sqrt{x}$ square.
There are $\binom{2 \sqrt{x}}{\sqrt{x}} = \exp(\Omega(x))$ monotone paths $\gamma $ through the grid,
and we wish to prove bounds on the maximum blue-red deviation $D$ achieved by any such path.

\paragraph{Gaining intuition: how blue-red deviations behave.} 
To gain intuition, let us start by considering the trivial path  $\gamma$ that contains
the entire grid beneath it. The expected number of blue dots (as well as red dots) 
beneath $\gamma$ is $\Theta(x)$, and the standard deviation on the number of blue dots
(as well as red dots) is thus $\Theta(\sqrt{x})$. With probability $\Omega (1) $, there are 
$\Omega(\sqrt{x})$ more blue dots in the grid than red dots, which results in a blue-red deviation
of $\Omega(\sqrt{x})$ for $\gamma$.

Of course, that's just the blue-red deviation of a \emph{single fixed path}. What should we expect
the \emph{maximum} blue-red deviation $D$ over all paths to be? On one hand, there are exponentially many paths that we must consider, but on the other hand,
the blue-red deviations of the paths are closely correlated to one another.
The result, it turns out, is that $ D $ ends up being an $x^{o(1)}$ factor larger than $\sqrt{x}$.

We will show that, with probability $1 - o(1)$, the maximum blue-red deviation $D$
is between $\Omega(\sqrt{x \log \log x})$ and $\tilde{O}(\sqrt{x})$. If we backtrack to our original problem (i.e., we relate the blue-red deviations to the crossing numbers, the crossing numbers to the displacements, and the displacements to the running times), we get that the amortized cost of insertions is between $\Omega(x^{1.5} \sqrt{\log \log x})$ and $\tilde{O}(x^{1.5})$. 

\paragraph{An upper bound of \boldmath $\tilde{O}(\sqrt{x})$ on the maximum blue-red deviation $ D $.}
The first step in bounding $ D $ is to prove a general result about decompositions of monotone paths. Let $k = 4\sqrt{x}$ be the perimeter of the grid. We claim that for any monotone path $\gamma $ through the grid,
it is always possible to decompose the area
under $\gamma $ into disjoint rectangles $R_1, R_2, \ldots$
such that the sum of the perimeters of the rectangles is at most
$k \log k$. 

Such a decomposition can be constructed recursively as follows:
(1) find the point $ q $ halfway along the path, and drop a rectangle
from $ q $ to the bottom-right-most point in the grid; (2) 
then recursively construct a rectangular decomposition for the
portion of the path prior to $ q $, and recursively construct a rectangular decomposition for the portion of the path
after $ q $. (See Figure~\ref{fig:path3} for an example.)

The recursive decomposition is designed so that, in the $ i$-th
level of recursion, each recursive subproblem takes place
on a grid with perimeter exactly $k / 2^i$. Since there are at most $2^i$ 
subproblems in each level of recursion, each of which contributes a rectangle with perimeter at most $k / 2^i$,
the sum of all the rectangle perimeters over all levels of recursion is at most $k \log k$. 

The next step in the analysis is to consider the maximum amount that any given
rectangle can contribute to the blue-red deviation of a path. 
If a given rectangle has area $a$, then the expected number of
blue/red dots in the rectangle is $\Theta(a)$, and
with high probability in $x$, the rectangle as a whole has
blue-red deviation $O(\sqrt{a} \log x)$.
This, in turn, means that if a rectangle has perimeter $p$,
then with high probability in $x$, it has blue-red deviation at most
$O(p \log x)$. Finally, since there are only $O(x^2)$ possible rectangles in the entire grid, and this property holds for each of them with probability $1 - 1/\poly(x)$,
the property also holds simultaneously for all of them with probability $1 - 1 / \poly(x)$.

Putting the pieces together, we know that every path $\gamma $ has
a rectangular decomposition such that the sum of the rectangle perimeters is
$O(\sqrt{x} \log x)$. We further know that, if a rectangle in the decomposition has perimeter
$p$, then it contributes at most $O(p \log x)$ to the blue-red deviation of $\gamma$.
It follows that the total blue-red deviation of any path $\gamma $ is at most $O(\sqrt{x} \log^2 x)$.

\paragraph{A lower bound of \boldmath $\Omega(\sqrt{x\log \log x})$ on the maximum blue-red deviation $ D $.}
Now we turn our attention to proving a lower bound on $ D $.
We wish to find a monotone path $\gamma$
whose blue-red deviation is $\Omega(x^{1.5} \sqrt{\log \log x})$.

Call a $ j\times j $ square within the grid \defn{high-value} if it has blue-red
deviation at least $\Omega(j\sqrt{\log \log x})$. The definition is designed so that 
every square has probability at least $1 / \sqrt{\log x}$ of being high-value. 

We construct a monotone path $\gamma $ recursively as follows.
First break the grid into quadrants, and check whether the top left quadrant is 
a high-value square. If so, then return the trivial path
that contains the entire grid below it.  Otherwise, recursively construct a path through the
bottom-left quadrant, recursively construct a path through the top-right quadrant,
and set $\gamma $ to be the concatenation of the two paths.\footnote{When we recurse,
we do not change the threshold $\Omega(j \sqrt{\log \log x})$ that dictates whether a given $j \times j$ square is special;
that is, $x$ acts as a global variable setting this threshold.} (For an example, see Figure~\ref{fig:path4}.)

There are two types of base cases in the recursion. The first type is when
a subproblem finds a high-value square; we call this a \defn{successful base case}. The second
type is when a subproblem terminates because it is on a $ 1\times 1 $ grid;
we call this a \defn{failed base case}. 

If a successful base case takes place on a sub-grid of width $w$,
then the high-value square that it discovers contributes $\Omega(w \sqrt{\log \log x})$
to the total blue-red deviation of $\gamma $. In order to establish
a lower bound of $\Omega(\sqrt{x\log \log x})$ on the blue-red deviation,
it therefore suffices to show that the sum of the widths of the successful base cases is 
$\Omega(\sqrt{x})$.\footnote{There is also a large portion of the area underneath $\gamma$ that is not contained in any of the high-value squares of the subproblems. Technically, we must also ensure that the blue/red dots in this unaccounted-for area do not substantially change the blue-red deviation, but this follows from a straightforward Chernoff bound.}

By construction, the sum of the widths of both the successful base cases and the failed base cases is 
exactly $\sqrt{x}$. Moreover, each failed base case has width exactly $ 1 $. Thus our task reduces
to bounding the number of failed base cases by $o(\sqrt{x})$ with probability $ 1 - o (1) $.

In order for a given failed base case to occur, there is a recursion path of $\Theta(\log x)$
sub-problems that must all fail to find a high-value square. Each of these failures occurs with
 probability at most $1 - 1 / \sqrt{\log x}$, so the probability of all of them occurring is
$$\left(1 - 1 / \sqrt{\log x}\right)^{\phantom{x} \Theta(\log x)} = o(1).$$
Since the probability of any given failed base case occurring is $o(1)$, 
the expected number of failed base cases that occur is $o(\sqrt{x})$.
By Markov's inequality, the number of failed base cases is $o(\sqrt{x})$ with probability $ 1 - o (1) $,
as desired.

\subsection{Stronger primary anti-clustering with larger rebuild windows}
\label{sec:stonger-primary-clustering}

In this section, we consider what happens if a larger rebuild window 
size  $R = n / \polylog (x)$ is used. As discussed in the introduction, this allows
for us to improve our amortized insertion time from $\tilde{\Theta}(x^{1.5})$ to 
$\tilde{\Theta}(x)$, while still achieving average query time $\Theta(x)$.

Remarkably, these bounds hold not just for hovering workloads, but also for arbitrary workloads that
stay below a load factor of $1 - 1/x$. To simplify discussion in this section,
however, we continue to focus on the hovering case.

There are two main technical challenges
that our analysis must overcome. The first challenge is obvious: we must quantify the degree to which
tombstones left behind by deletions improve the performance of subsequent insertions.
The second challenge is a bit more subtle: in order to support large rebuild-window sizes $R$,
our analysis must be robust to the fact that tombstones can accumulate over time, increasing the effective
load factor of the hash table. 
This latter challenge is further exacerbated by the fact that the choice of which tombstones are in the table
at any given moment is a function not only of the sequence of operations being performed, but also of the
randomness in the hash table. 
This means that, even if the cumulative load factor from the elements and tombstones can be bounded (e.g., by $ 1 -\Theta(1/x)$), we still cannot analyze the tombstones as though they were normal elements; a consequence of this is that we cannot even apply the classic analysis to deduce an $O(x^2)$-time bound for insertions or an $O(x)$-time bound for queries. 

\paragraph{Using crossing numbers to rescue the queries.}
Let us consider the time that it takes to query an element whose hash is $ j $.
The time is proportional to the number of elements and tombstones that have hashes
smaller than $j$ but that reside in positions $j$ or larger.\footnote{Technically, we must also
consider elements that hash to exactly $j$ but the expected number of such elements is $O(1)$.}
Call the slots containing these elements \defn{$j$-crossed}. 
Right after a rebuild is performed, the number $ s $ of $j$-crossed
slots has expected value $O(x)$. Over the course of the time window between consecutive rebuilds, however,
the quantity $ s $ gradually increases. 

Fortunately, the amount by which $ s $ increases is \emph{precisely}
the crossing number $c_j$. Indeed, $ c_j $
gets incremented exactly whenever a formerly non-$ j $-crossed slot
becomes $ j $-crossed.

Thus, if we can bound $c_j$ to be small then we hit two birds with one stone:
we are able to bound the running times of both queries and insertions.

\paragraph{But how do we rescue the crossing numbers?}
Large rebuild windows also break the analysis of crossing numbers, however.
In the original analysis, we argued that there is most likely some position $j' < j$
satisfying $j - j' = O(x^2)$ such that position $j'$ remains an empty slot throughout
the entire rebuild window. This meant that, when considering $c_j$,
we only had to analyze insertions and deletions that hash into the interval $I = [j' + 1, j)$
of size $\Theta(x^2)$.

We can no longer argue that such a $j'$ necessarily exists, however, since the accumulation of tombstones over time might eliminate all of the free slots near position $j$. 
Thus we must 
extend our analysis to consider intervals $I$ of size $\omega(x^2)$.

Fortunately, if we consider any interval $I$ of size at least $x^2 \polylog x$,
then we can argue that the interval most likely \emph{initially} contains
$\Omega(|I| / x)$ free slots. Define the \defn{insertion surplus} of an interval $ I $
to be the maximum blue-red deviation of any monotone path through the grid representing $I$,
\emph{minus} the number of free slots initially in $I$. We prove that the crossing number
$c_j$ is exactly equal to the maximum insertion surplus of any interval $I$ of the form $[j' + 1, j)$.
Since large intervals $I$ have a $\Theta(1/x)$-fraction of their slots initially empty,
it is very unlikely that they end up determining the crossing number $c_j$. The result is that we
can again focus primarily on intervals of size $O(x^2)$, and perform the analysis of $ c_j $ as before.

\paragraph{Putting the pieces together for \boldmath $R = n / \polylog (x)$.}
We now analyze the case of
$R = n / \polylog (x)$. As before, we use the displacement $d_u$ for an insertion as a proxy for insertion time
(although bounding the difference between the two requires a more nuanced argument
than  before, see Lemmas~\ref{lem:runningtimedecomposition} and~\ref{lem:positionaloffset}). We can then relate the displacements to the crossing numbers by
$$\sum_u d_u = \sum_{j \in [n]} c_j.$$
Now, however, both sums consist of between $n / \polylog x$ and $n$ terms.
This means that for a random insertion $ u $ and a random $ j\in [n] $,
$$\E[d_u] \le \E[c_j \polylog(x)].$$
As before, we can transform the problem of bounding $ c_j $ into the 
problem of bounding the maximum blue-red deviation of any monotone path
in a certain grid. The expected number of blue/red dots in the grid is 
now $x^2 / \polylog(x)$ (rather than $\Theta(x)$). Thus our bound on blue-red deviation
comes out as $\tilde{O}(x / \polylog (x)) = O(x)$ (rather than $\tilde{O}(\sqrt{x})$).
This means that $\E[c_j] = O(x)$, which implies that the
 expected time taken by any query is also $O(x)$ and that the average time taken by 
each insertion is $\tilde{O}(x)$. 

We can now also see why $ n/\polylog (x) $ is the right rebuild window size to use. In particular, if we make $ R $ smaller than $ n/\poly (x) $, then insertion times suffer, but if we let $ R $ get too close to $ n $ (or exceed $ n $) then the crossing numbers $ c_j $ become $\omega (x) $ (thanks to our lower bound construction on blue-red deviations), and thus queries take time $\omega(x)$. Thus it is impossible to select a value for $R$ that achieves expected time $O(x)$ for all operations, and if we want $O(x)$-time queries, we must make $R = o(n)$.

\paragraph{Analyzing arbitrary workloads.}
Finally, we generalize these results to \emph{any} sequence of operations that stays below a load factor of $1 - 1/x$. 
We argue that, within any rebuild window, it is possible to re-organize the operations in such a way that (a)~none of the crossing numbers decrease, 
and (b)~the operations consist of a series of insertions, followed by a series of alternating insertions/deletions, followed by a series of deletions (see Proposition~\ref{prop:amortized}).
The alternating insertions/deletions can be analyzed as above, and because $R$ is large, the crossing-number cost of the initial insertions can be amortized away.
The fact that the re-organization of operations does not decrease any crossing numbers ends up being easy to prove using our characterization of crossing numbers in terms of insertion surpluses of intervals.

\subsection{Graveyard hashing}

Can we achieve a bound of $O(x)$ for both query and insertion times? We have already seen that standard ordered linear probing cannot, regardless of the choice of $R$. 

The guarantee can be achieved, however, by a slightly modified version of ordered linear probing
that we call graveyard hashing. The basic idea of graveyard hashing is to
insert extra tombstones artificially during each rebuild. This extra injection of tombstones ensures that, during the next rebuild window, tombstones for insertions to make use of remaining plentiful at all times. Thus, we are able to put all insertions in a ``best-case scenario'', guaranteeing that each operation has expected time $O(x)$. For details, see 
Section~\ref{sec:graveyard}.

\begin{figure}[h]
        \centering
        \begin{subfigure}[t]{0.3\textwidth}
            \centering
            \includegraphics[trim={9cm 42cm 24cm 12cm},clip,width=\textwidth]{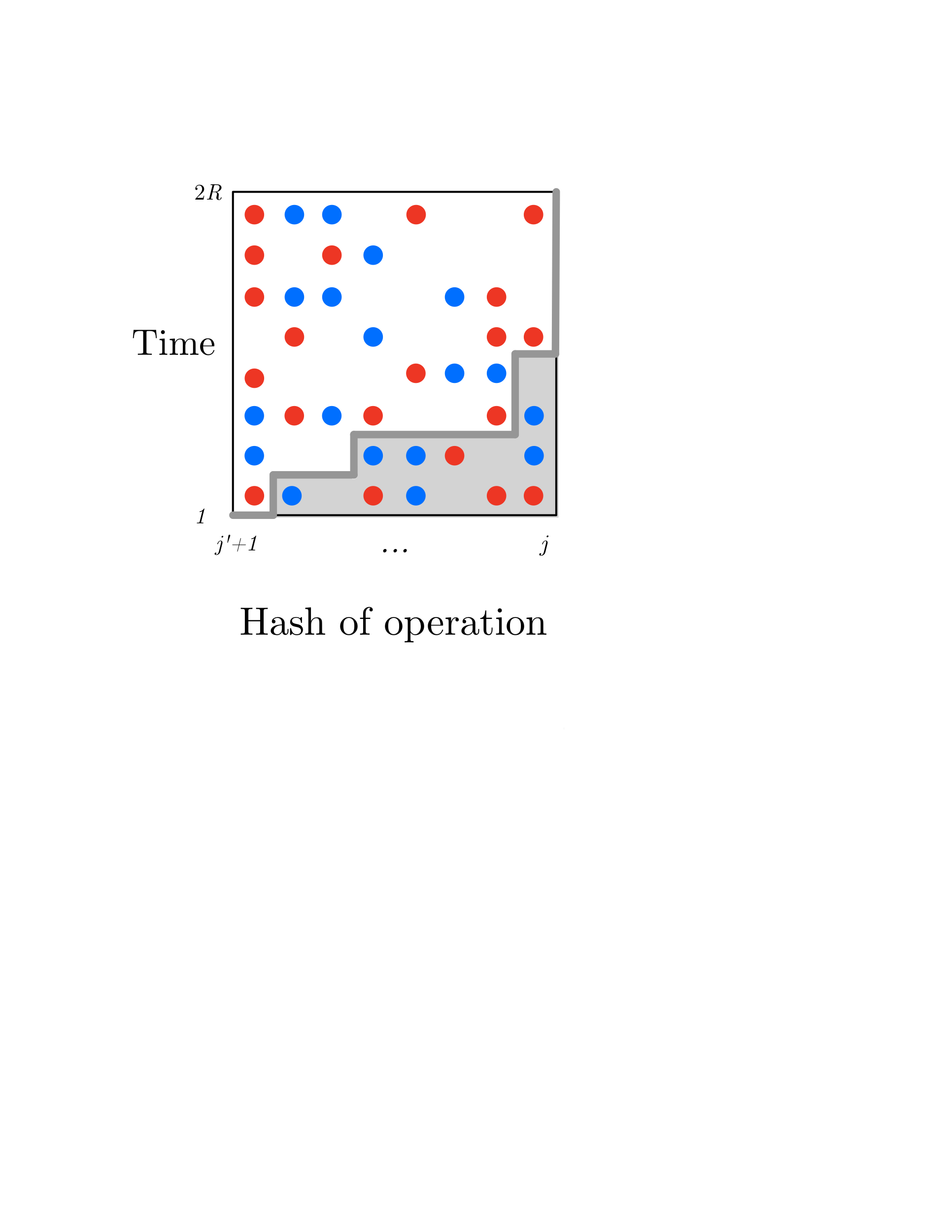}
            \caption{An example of representing insertions/deletions in $I = [j' + 1, j)$ graphically over time.
            Blue dots are insertions and red dots are deletions. (To simplify this picture, we plot each point in one of eight $y$-coordinates, but in reality, each dot would have a distinct $y$-coordinate.) An example monotone path is also given, and the 
            blue-red deviation is $6 - 4 = 2$.}    
            \label{fig:path1}
        \end{subfigure}
        \hfill
        \begin{subfigure}[t]{0.3\textwidth}  
            \centering 
            \includegraphics[trim={9cm 42cm 24cm 12cm},clip,width=\textwidth]{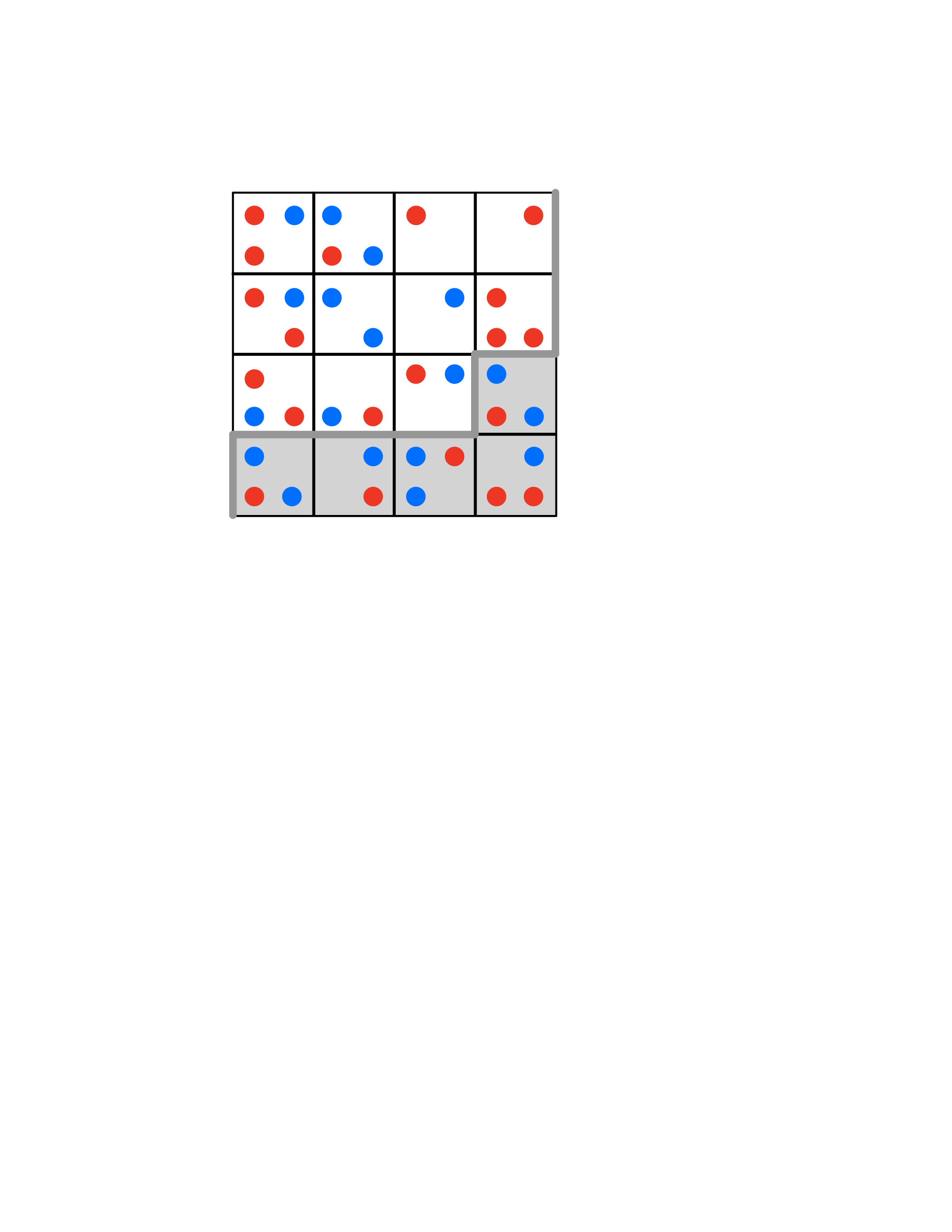}
            \caption{To simplify the problem, we draw a $\sqrt{x} \times \sqrt{x}$ grid, and consider
            only paths that go along the drawn grid lines. Here $x = 16$.}    
            \label{fig:path2}
        \end{subfigure}
        \hfill
        \begin{subfigure}[t]{0.3\textwidth}   
            \centering 
            \includegraphics[trim={10cm 42cm 24cm 12cm},clip,width=\textwidth]{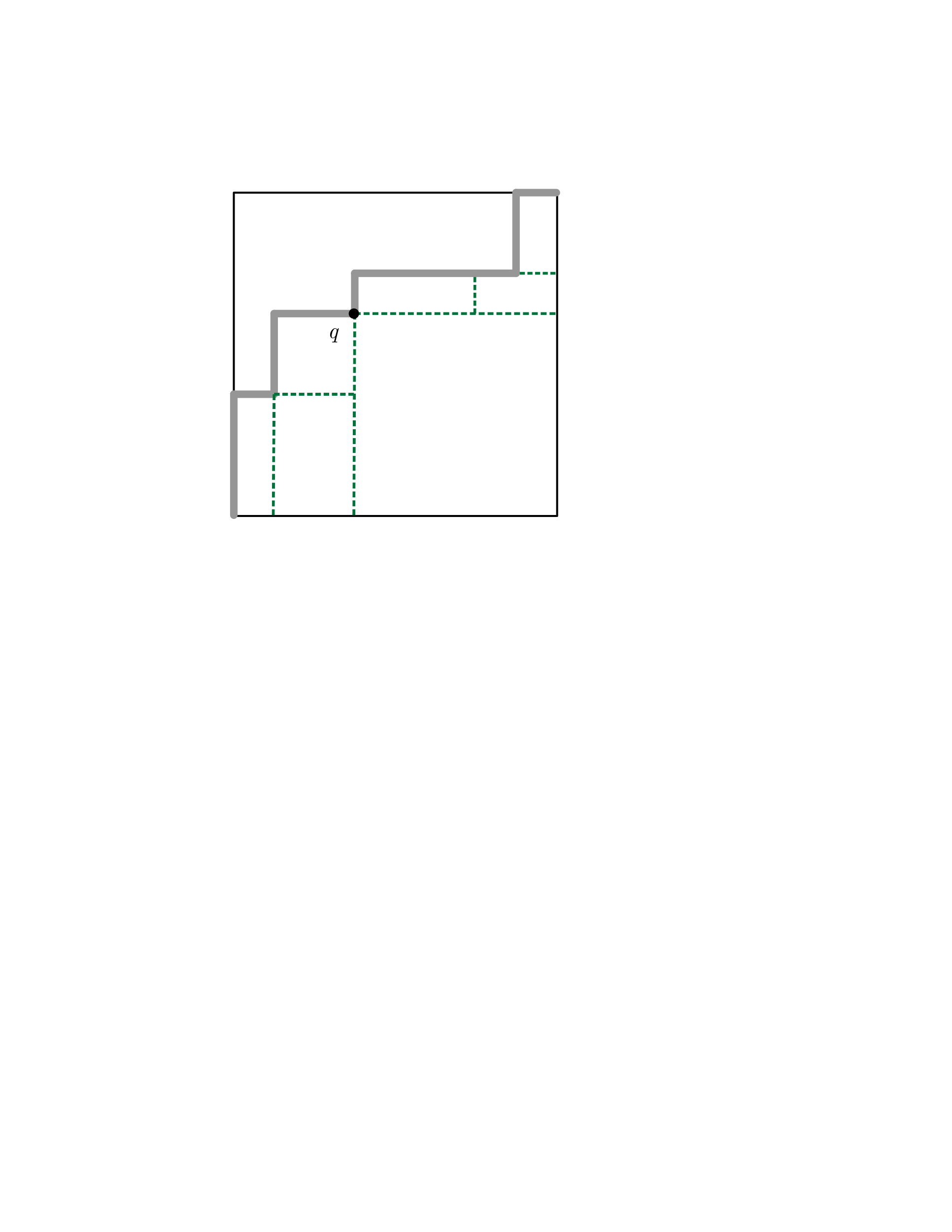}
            \caption{An example of a recursively constructed rectangular decomposition for a path. The point $q$
            used in the top level of recursion is also labeled.}
            \label{fig:path3}
        \end{subfigure}
        \newline
        \hfill
        \begin{subfigure}[t]{0.8\textwidth}   
            \centering 
            \includegraphics[trim={10cm 42cm 24cm 12cm},clip,width=0.4 \textwidth]{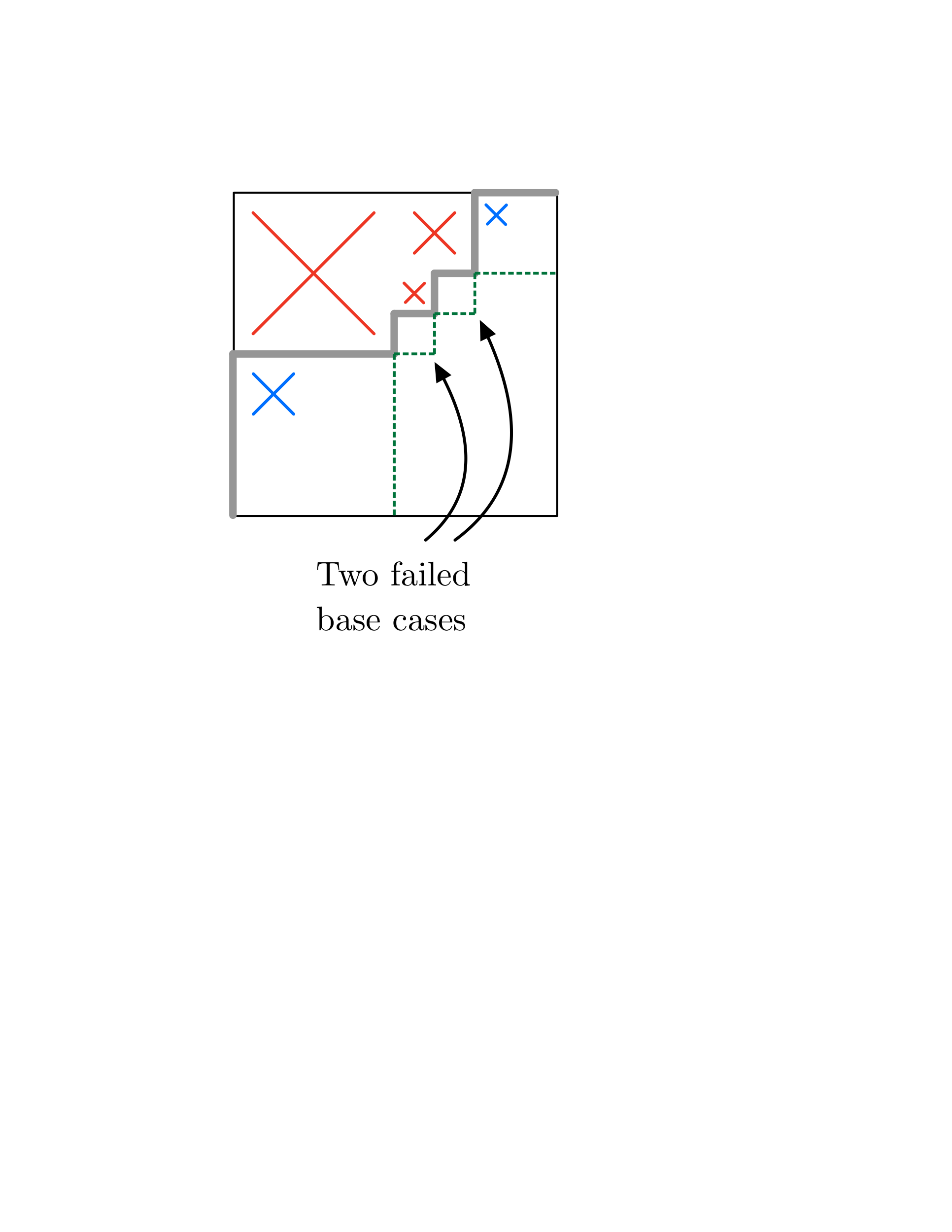}
            \caption{An example of the construction for a path with high blue-red deviation. 
            For each recursive subproblem whose top-left quadrant is not high-value, we place a 
            red X through the quadrant; for each recursive subproblem whose top-left quadrant is high-value, we
            place a blue X through the quadrant; the only subproblems not to have an X drawn in them are the failed base cases, each of which takes place on a $1 \times 1$ grid. There are four base cases, two of which are failed.}    
            \label{fig:path4}
        \end{subfigure}
        \caption{}
        \label{fig:paths}
    \end{figure}



\section{Some Basic Balls-and-Bins Lemmas}

We begin by proving several basic lemmas having to do with balls and bins. What makes these lemmas different from standard balls-and-bins bounds is that they consider all prefixes of a sequence of bins, bounding the probability of any prefix behaving abnormally. 
In order to demonstrate how these lemmas relate to linear probing, we will also use them to reprove
the classic bounds on the performance of insertions and queries for ordered linear probing.

Throughout the rest of the section, consider the setting in which
 we place $ m =\Theta (n) $ balls randomly into $ n $ bins.
Let $\mu = n/m $ be the expected number of balls in each bin.

\begin{lemma}
Let $x > 1$ and $ k \ge 1 $. With probability $ 1-2 ^ {-\Omega (k) } $, for every $ i \ge x^{2 } k $,
the number of balls in the first $ i $ bins is between $ (1 - 1/x) i\mu $ and $ (1 + 1/x) i\mu $.
\label{lem:saturated}
\end{lemma}
\begin{proof}
Let $\ell = x^{2 } k$. We wish to show that, with probability at least $ 1-2 ^ {-\Omega (k) } $,
there is no $i \ge \ell$ such that the first $ i $ bins contain either fewer than $ (1 - 1/x) i\mu $ balls
or more than $ (1 + 1/x) i\mu $ balls. We will focus on the more-than-$ (1 + 1/x) i\mu $ case,
since the fewer-than-$ (1 - 1/x) i\mu $ case follows by a symmetric argument.

Let $ X_i $ be the indicator random variable for the event that the first $ i $ bins contain at least $ (1 + 1/x) i\mu $ balls.
By a Chernoff bound, 
$$\Pr[X_\ell] \le 2 ^ {-\Omega (k) }.$$

It is tempting to apply a similar Chernoff bound to every $  i \ge \ell$, and then to take a union bound;
but this would yield a bound of $\Pr[X_i \text{ for some } i \ge \ell] \le \ell 2^{-\Omega(k)}$ rather than the desired
bound of $\Pr[X_i \text{ for some } i \ge \ell] \le 2^{-\Omega(k)}$. Thus a slightly more delicate approach is needed.

To complete the proof, we prove directly that
\begin{equation}
\Pr[X_i \text{ for some } i \ge \ell] \le O(\Pr[X_{\ell}]),
\label{eq:nounion}
\end{equation}
which we have already shown to be $2 ^ {-\Omega (k) }.$

Suppose that $X_i$ holds for some $ i \ge \ell$ and let $ j $ be the largest such $ i $.
The number $B$ of balls in the first $ j $ bins must satisfy $B \ge (1 + 1/x) j\mu $. Each of these $B$ balls independently has
probability $\ell /j $ of being in one of the first $\ell$ bins. Conditioning on a given value of $B$, the number of balls in the first
$\ell$ bins is a binomial random variable with expected value at least $(1 + 1/x) \ell \mu$. 
With probability $\Omega(1)$ (and using the fact that $(1 + 1/x) \ell \mu = \Omega(1)$), this binomial random variable takes a value greater than or equal to its mean $(1 + 1/x) \ell \mu$,
which implies that $X_\ell$ occurs. Having established \eqref{eq:nounion}, the proof is complete. 
\end{proof}

\begin{corollary}
Consider the largest $ i $ such that the first $i$ bins contain at least $ (1 + 1/x) i\mu $ balls.
Then $\E[i] \le O(x^{2})$.
\end{corollary}

\begin{lemma}
Let $x > 1 $ and $ k \ge 1 $. With probability $ 1-2 ^ {-\Omega (k) } $, there does not exist any $ i $ such that
the first $ i $ bins contain at least $(1 + 1/x) i\mu + k x $ balls.
\label{lem:bolus}
\end{lemma} 
\begin{proof}
Let $Y_i$ be the indicator random variable for the event that the first $ i $ bins contain at least $(1 + 1/x) i\mu + k x $ balls.
Let $ r = x^{2} k$. By Lemma \ref{lem:saturated},
$$\Pr[Y_i \text{ for some }i \ge r] \le 2^{-\Omega(k)}.$$
Thus it suffices to argue that
$$\Pr[Y_i \text{ for some }i \le r] \le 2^{-\Omega(k)}. $$
As in the previous proof, taking a Chernoff bound and then summing over $ i $ will not give the
result that we are aiming for. Thus a more refined approach is again needed.

Suppose that $ Y_i $ occurs for some $ i \le  r$, and let $ j $ be the largest such $ i $.
Let $B$ satisfying $B \ge j \mu + k  x$ be the number of balls in the first $ j $ bins, and let
 $2 ^ q $ be the largest power of two satisfying $2^q \le j$. Each of the $ B $ balls in the first $ j $ bins
has probability $ 2 ^ q/j $ of being in the first $ 2 ^ q $ bins. Conditioning on a given value of $ B $,
we therefore have that the number of balls that land in the first $ 2 ^ q $ bins is a binomial random variable with expected value
at least $2^q \mu + k x / 2$. With probability $\Omega (1) $, this random variable takes a value at least as large as
its mean. That is, if we condition on $ Y_i $ occurring for some $ i \le r $, then with constant probability there is some power 
of two $2^q \le r$ such that at least $2 ^ q\mu + kx / 2 $ balls land in the first $ 2 ^ q $ bins.

For $ q \in \{0, 1, \ldots, \log r\}$, let $ Z_q $ denote the indicator random variable for the event that at least
$2^q\mu + kx / 2$ balls land in the first $ 2 ^ q $ bins. So far, we have shown that
$$\Pr[Y_i \text{ for some }i \le r] = O(\Pr[Z_q \text{ for some }q \le \log r]). $$
By a Chernoff bound, if we consider a given $Z_q$, and we define $s$ such that  $2^{q + s} = kx$, then
$$
\Pr[Z_q] \le 
\begin{cases}
2^{-\Omega(skx)} & \text{ if } s > 0, \\
2^{-\Omega\left(k^2x^2 / 2^q \right)} & \text{ if } s \le 0.\\
\end{cases}
$$
Recalling that $r = x^2k$, it follows that
\begin{align*}
    \Pr[Z_q \text{ for some }q \le \log r] & \le \sum_{s > 0} 2^{-\Omega(skx)} +
\sum_{q = \log (kx)}^{\log (kx^2)} 2^{-\Omega\left(k^2x^2 / 2^q \right)}.
\end{align*}
The first sum is a geometric series summing to $2^{-\Omega(kx)}$. The second sum is dominated
by its final term which is $2^{-\Omega(k)}$. Thus the lemma is proven.
\end{proof}

To demonstrate how these results relate to linear probing, we now use them to re-create the classic upper bounds on insertion performance (Proposition \ref{prop:classic}) and query performance (Proposition \ref{prop:classic2}) for ordered linear probing.
\begin{proposition}
Starting with an empty linear-probing hash table of $ n $ slots, suppose that we perform $(1 - 1/x) n$ insertions.
The running time $ T $ of the final insertion satisfies
$$\Pr[T \ge kx^{2 }] \le 2 ^ {-\Omega (k) }. $$
\label{prop:classic}
\end{proposition}
\begin{proof}
If the final insertion hashes to some position $ p $, and takes time $ T \ge kx^{2 }$,
then the insertion must have been inserted into a run of elements going from position $p - i$
 to position $ p + T - 1 \ge P + kx^2 - 1$ for some $ i $. Consequently, the number of elements $u$ that satisfy $h(u) \in [p - i , p + kx^2  - 1]$ must be at least $i + kx^2$, even though the expected number of such elements is $(1 - 1/x) (i + kx^2)$.
By treating the positions $p + kx^2 - 1, p + kx^2 - 2, \ldots$ as bins, we can apply Lemma \ref{lem:saturated}
to deduce that the probability of such an $ i $ existing is at most $2 ^ {-\Omega (k) }. $
\end{proof}

\begin{proposition}
Consider an ordered-linear-probing hash table with $ n $ slots that contains $(1 - 1/x) n$ elements and no tombstones,
and suppose we perform a query.
The running time $ T $ of the query satisfies
$$\Pr[T \ge kx] \le 2 ^ {-\Omega (k) }. $$
\label{prop:classic2}
\end{proposition}
\begin{proof}
Notice that we do not have to distinguish between positive and negative queries, since even for a negative query
we only need to perform a linear scan until we find a key with a larger hash than the one we are querying.

If the key that we are querying hashes to some position $ p $, and takes time $ T  \ge kx$,
then all of the elements in positions $p, \ldots, p + T - 2 \ge p + kx - 2$ must have hashes
at most $p$. Consider the largest $ i $ such that all of positions $p - i, \ldots, p + kx - 2$
contain elements. All of the elements in positions $p - i, \ldots, p + kx - 2$ must have hashes between $p - i$ and $p$. Thus the total number of elements that hash to positions $[p - i,  p]$
must be at least $i + kx - 1$, even though the expected number of elements that hash to those positions is $(1 - 1/x) (i + 1)$.
By treating the positions $p, p - 1, p - 2, \ldots$ as bins, we can apply Lemma \ref{lem:bolus} to deduce that
the probability of any such $ i $ existing is at most $2 ^ {-\Omega (k) }$.
\end{proof}

The proof of Proposition \ref{prop:classic2} comes with a corollary that will be useful to reference later.
\begin{corollary}
Consider an ordered-linear-probing hash table with $ n $ slots that contains $(1 - 1/x) n$ elements and no tombstones. Consider a position $ i $, and let $T$ be the number of elements $u$ that reside in positions $i$ or greater, but that have hashes $h(u) < i$. Then 
$$\Pr[T \ge kx] \le 2 ^ {-\Omega (k) }. $$
\label{cor:over}
\end{corollary}

\section{Bounds on Insertion Surplus}\label{sec:surplus}

In this section, we introduce two core technical propositions that will be used in subsequent sections 
for our analysis of ordered linear probing. 

Consider a sequence $S$ of operations which alternate between insertions and deletions. (Think of $S$ as the operations
between two rebuilds.)
Let $n$ be the number of slots in the hash table. Let $ P $ be a sub-interval of $[n]$ and 
let $S_P$ denote the subset $\{u \in S \mid h(u) \in P\}$.

We say that a subset $ S' \subseteq S_P$ is \defn{downward-closed} (with respect to $S_P$) if it satisfies the following
property: for every insertion or deletion $ u \in S'$, every $ v \in S_P$ that occurs temporally before $ u $
and satisfies $h(v) \ge h(u)$ is also in $ S' $. Define the \defn{insertion surplus}  of a 
downward-closed set $ S' $ to be the number of insertions in $ S' $
minus the number of deletions in $ S' $, if there are more insertions than deletions in $ S' $,
and $ 0 $ otherwise. 

The purpose of this section is to prove upper and lower bounds on
the maximum insertion surplus of any downward-closed subset $ S' \subseteq S_P $. We will parameterize these bounds by $\mu := \E[|S_P|] / 2$, which is the expected number of insertions (and also the expected number of deletions) in $S_P$.

Notice that $ S_P $ itself is downward-closed. If we assume that all of the operations in $S_P$ are on distinct keys, then the expected insertion surplus of $ S_P $ will be
$\Omega (\sqrt { \mu }) $ (since the number of insertions and the number of deletions in $ S_P$ both have standard deviation $\Theta (\sqrt { \mu }) $).
A natural question is whether there exists any downward-closed subset $ S' \subseteq S_P$ with a significantly larger insertion surplus. 

This section proves two propositions. 
\begin{proposition}
Suppose that $|P| \ge \sqrt{\mu}$. 
With probability $ 1 - 1/\poly (\mu) $, every downward-closed subset $ S'\subseteq S_P $ has insertion surplus
$\tilde { O } (\sqrt { \mu }) $.
\label{prop:surplus}
\end{proposition}

\begin{proposition}
Suppose that the insertions/deletions in $ S_P$ are all on different keys. Suppose that $|S| \le n$, that $n$ is sufficiently large as a function of $ \mu $ and $|P|$, and that $|P| \ge \sqrt{\mu}$. 
Then with probability $ 1 - o (1) $,
there exists a downward-closed subset $ S'\subseteq S_P $ with insertion surplus
$\Omega (\sqrt { \mu\log\log \mu }) $.
\label{prop:surpluslower}
\end{proposition}

Whereas Proposition~\ref{prop:surplus} tells us that no downward-closed subset $ S' $ has significantly larger insertion surplus
then the expected insertion surplus of $ S_P $, Proposition~\ref{prop:surpluslower} tells us that there most likely is some downward-closed subset $ S' $ whose insertion surplus
is (at least slightly) asymptotically larger than the expected insertion surplus of $ S_P$.

\paragraph{Reinterpreting insertion surplus in terms of paths on a grid.}
Before proving the propositions, we first reinterpret the propositions in terms of paths on a grid.
Consider a grid $G$ with height $|S|$ and width $|P|$. Define $ r $ so that $ r +1 $
is the beginning of interval~$P$. Color the cell $(i, j)$ in the grid $ G $
blue if the $i$-th operation in $S$ is an insertion whose hash is $ r + j $;
and color the cell $ (i, j) $ red if the $i$-th operation in $S$ is a deletion whose hash is
$ r + j $. (Many cells will be neither red nor blue.) 
The blue and red cells in the grid $G$ correspond to the insertions and deletions in $ S_P$.

Now consider the set of \defn{monotone paths through $G$}, that is, the set of paths that begin in the bottom left corner,
 and then walk along grid lines to the top right corner, only ever moving up or to the right.
Define the \defn{blue-red differential } of a path to be the number of blue cells that reside below the path minus the number of red cells that
reside below the path (we say that the path \defn{covers} these cells).

A subset $ S' \subseteq S$ is downward-closed if and only if there is a monotone path $\gamma$ 
through $ G $ such that the blue and red cells covered by $\gamma $ in $ G $ are precisely 
the insertions and deletions in $ S' $. 
Thus, rather than considering downward-closed subsets $S'$ of $ S_P$, we can consider monotone paths $\gamma $
through $ G $, and rather than considering the insertion surplus of each subset $ S' $, we can consider the blue-red differential
of each monotone path $\gamma $. Although this distinction may at first seem superficial, we shall see later that
the geometry of monotone paths makes them amenable to clean combinatorial analysis.

\paragraph{Considering a more coarse-grained grid \boldmath $G'$.}
One aspect of $ G $ that makes it potentially unwieldy is that it will likely be sparse, 
meaning that the vast majority of cells are neither red nor blue. 
Thus, in our proofs, it will also be useful to define a more coarse-grained grid $ G' $
that is laid on top of $G$.
The grid $ G' $ has width and height $\sqrt{\mu}$, meaning that each cell in $ G' $
corresponds to a sub-grid of $ G $ with width $|P| / \sqrt{\mu}$ and height $|S|/ \sqrt{\mu}$.
To avoid confusion, we will refer to the blue/red cells in $ G $ as blue/red dots in $ G' $.

Note that, since $G'$ is a $\sqrt{\mu} \times \sqrt{\mu}$ grid, we are implicitly assuming
that $\sqrt{\mu}$ is a positive integer; this assumption is w.l.o.g.. 
To simplify our discussion (so that
we can treat all of the cells of $G'$ as having uniform widths and heights),
 we will further assume that $|P|$ and $|S|$ are divisible by $\sqrt{\mu}$.
These assumptions can easily be removed by rounding all of the quantities to
powers of four, and performing the analysis using the rounded quantities.\footnote{Importantly, both propositions assume $|P| \ge \sqrt{\mu}$, so if we round both $|P|$ and $\mu$ to powers of $4$, then the rounded value for $|P|$ will be a multiple 
of the rounded value for $\mu$ (and $\mu$ will be a square number).}

The grid $G'$ is parameterized so that the expected number of blue dots (resp. red dots)
in each cell is exactly $1$. As terminology, we say that for each cell in $G'$, the insertions and deletions that \defn{pertain}
that cell are the ones that have blue/red dots in that cell; and the keys that \defn{pertain} to the cell are the keys
that have at least one insertion/deletion pertaining to the cell.  The expected number
of distinct keys that pertain to a given cell in $G'$ is at most $1$. Moreover, by a Chernoff bound, 
and with probability $1 - 1/\poly(\mu)$, each cell has at most $O(\log \mu)$ keys that pertain to it.

\subsection{Proof of Proposition \ref{prop:surplus}}

We wish to show that, with probability $1 - 1/\poly(\mu)$,
every monotone path $\gamma $ through $ G $ has blue-red differential at most $\tilde{O}(\sqrt{\mu})$. 
The next lemma establishes that, rather than considering monotone paths in $ G $,
it suffices to instead consider monotone paths in $ G' $. 
\begin{lemma}
With probability $1 -1 /\poly (\mu)$, the following holds. For every monotone path $\alpha $ with blue-red differential $ a $ in $ G $, 
there is a monotone path $\beta $ with blue-red differential $ b $ in $ G' $ 
such that $a - b = O(\sqrt{\mu} \log \mu)$.
\label{lem:Gprime}
\end{lemma}
\begin{proof}
Define $\beta $ to be the same as $\alpha $, except that the path is rounded to the grid lines of $ G' $ as follows:
for any cell in $ G' $ that the path $\alpha $ goes through the interior of, we round the path so that $\beta $ does 
not cover any of the points from that cell.

Every cell in $ G' $ that is entirely covered by $\alpha $ is also entirely covered by $\beta $; similarly, every cell in $ G' $
that is entirely not covered by $\alpha $ is also entirely not covered by $\beta $.
Thus the only difference between $\alpha $ and $\beta $ is that there are $O(\sqrt{\mu})$ cells in $G'$
that are partially covered by $\alpha $ but that are not covered by $\beta $.

For each cell in $ G' $, define the \defn{risk potential} of that cell to be the maximum blue-red differential of any
monotone path in $ G $ from the bottom left corner of that cell to the top right corner of that cell. 
To complete the proof, it suffices to show that every cell in $ G' $ has risk potential at most $O(\log \mu)$.
Notice, however, that the risk potential of each cell is at most as large as the number of distinct keys that pertain to the cell.
Thus the risk potentials of the cells are all $O(\log \mu)$ with probability $1 - 1 / \poly(\mu)$.
\end{proof}

To complete the proof of Proposition \ref{prop:surplus},
it suffices to show that, with probability $1 - 1/\poly(\mu)$, every monotone path through $ G' $ has
blue-red differential at most $\tilde{O}(\sqrt{\mu})$. 

The rest of the proof is completed in two pieces. The first piece is to show that for every monotone path $\gamma $,
it is possible to decompose the area under the path into rectangles where the sum of the perimeters of the rectangles is 
$O(\mu \log \mu)$. The second piece is to show that, for each rectangle in $ G' $, the blue-red differential of that rectangle 
is at most the perimeter of the rectangle times $O(\log \mu)$. Combining these two facts, we get a bound on the
blue-red differential of any monotone path $\gamma $.

\begin{lemma}
Consider any monotone path $\gamma$ through an $a \times b $ grid, where $\ell = a + b $ is a power of two. The area under the path can be decomposed into disjoint rectangles
such that the sum of the perimeters of the rectangles is $O(\ell \log \ell)$.
\label{lem:perimeters}
\end{lemma}
\begin{proof}
We construct the rectangular decomposition through the following recursive process. 
Break the path $ \gamma $ into two pieces of length $ \ell /2 $ which we call $ \gamma_1 $ and $ \gamma_2 $,
and let $ q $ be the point at which the two pieces meet. Let $ R $ be the rectangle whose top left corner is $ q $,
and whose bottom right corner is the bottom right point of the grid. Then we define our rectangular decomposition
 to consist of the rectangle $ R $, along with a recursively constructed rectangular decomposition for the path $ \gamma_1 $, 
and a recursively constructed rectangular decomposition for the path $ \gamma_2 $. The two recursive decompositions take place 
in the grids containing $ \gamma_1$ and $ \gamma_2$, and the base case of the recursion is when we get to  path that
 is either entirely vertical or entirely horizontal (meaning that there is no area to decompose into rectangles).

By design, the $ i $-th level of recursion takes place on a grid with perimeter $O(\ell / 2^i)$. 
It follows that the rectangles added in the $ i $-th level of recursion each have perimeters $O(\ell / 2^i)$.
On the other hand, the $ i $-th level of recursion has at most $2^i$ recursive subproblems,
so the total perimeter of the rectangles added in those subproblems is at most $O(\ell)$. 
There are at most $O(\log \ell)$ levels of recursion, so the sum of the perimeters 
of all of the rectangles in the decomposition is at most $O(\ell \log \ell)$. 
\end{proof}

\begin{lemma}
Consider any rectangle $ X $ in grid $ G' $. If $ X $ has perimeter $ k $, then with probability $ 1 - 1 / \poly(\mu) $, 
the number of blue dots minus the number of red dots
in $ X $ is $ O (k\log \mu). $
\label{lem:rectangle}
\end{lemma}
\begin{proof}
Let $ T $ be the time window that $ X $ covers on its vertical axis. The total number of insertions that occur in $ T $ is the same as the number of deletions that occur in $ T $,
up to $\pm 1$. Call an insertion in $ T $ \defn{serious} if the key being inserted has not previously been deleted in the same time window, and call a deletion in $ T $ \defn{serious} 
if the key being deleted is not subsequently reinserted in the same time window. Notice that the number of blue dots from non-serious insertions in $ X $ is the same as the number of red dots from
non-serious deletions in $ X $, so we can ignore both. Since the number of non-serious insertions equals the number of non-serious deletions,
the number of serious insertions in $ T $ is the same (up to $\pm 1$) as the number of serious deletions in $ T $.

Since $ X $ has perimeter $ k $, it can contain at most $k^2 $ cells in $ G' $. Thus the expected number of distinct keys
that pertain to the cells in $ X $ is $ O (k^2) $. It follows that the expected number of serious insertions (and similarly, the expected number of serious deletions) 
of keys that pertain to $ X $ is $ O (k^2) $. In order for the number of blue dots minus the number of red dots in $ X $ to exceed $D = \Omega (k\log \mu) $,
we would need that either (a) the number of serious insertions pertaining to $X$ is at least $ D/2 $ greater than its mean; or (b) the number of serious deletions pertaining to
$ X $ is at least $ D/2 $ smaller than its mean. By a Chernoff bound, the probability of either (a) or (b) occurring is that most
$1 / \poly (\mu) $.\footnote{Note that all of the serious insertions (resp. serious deletions) are on distinct keys, meaning that their hashes are independent, hence the Chernoff bound.}
\end{proof}

\begin{proof}[Proof of Proposition \ref{prop:surplus}]
By Lemma \ref{lem:Gprime}, it suffices to show with probability $1 - 1/\poly(\mu)$ that every monotone path $\gamma $ through $ G' $ has blue-red differential $\tilde{O}(\sqrt{\mu})$. 
By Lemma \ref{lem:perimeters}, every such path has a rectangular decomposition where the sum of the perimeters of the rectangles is $\tilde{O}(\sqrt{\mu})$. 
By Lemma \ref{lem:rectangle}, with probability $1 - 1/\poly(\mu)$, the contribution of each of the rectangles to the blue-red differential of $\gamma $
is that most $O(\log \mu)$ times the perimeter of the rectangle. Thus, with probability $1 - 1/\poly(\mu)$, every monotone path $\gamma$ through
 $ G' $ has blue-red differential $O(\sqrt{\mu} \log^2 \mu)$.
\end{proof}

\subsection{Proof of Proposition \ref{prop:surpluslower}}

We wish to construct a monotone path through $G'$ that has blue-red deviation $\Omega(\sqrt{\mu \log \log \mu})$.

Let $a_{i, j}$ be the number of insertions that pertain to cell $(i, j)$ and let $b_{i, j}$
the number of deletions that pertain to cell $(i, j)$ in $G'$. 
Each $ a_{ i, j } $ and each $b_{i, j}$ is a binomial random variable with mean $1$.
However, the random variables are not completely independent (they are slightly negatively correlated), which makes them a bit unwieldy to work with.
To handle this, the following lemma Poissonizes the random variables in order to show that they are $o(1)$-close to being independent.

\begin{lemma}
Let $A$ be the random variable $\langle a_{i, j}, b_{i, j} \mid i,j \in [\sqrt{\mu}] \rangle$
and let $\overline{A}$ be a random variable $\langle \overline{a}_{i, j}, \overline{b}_{i, j} \mid i, j \in [\sqrt{\mu}] \rangle$,
where each $\overline{a}_{i, j}$ and each $\overline{b}_{i, j}$ is an independent Poisson random variable with mean
$ 1 $. Then it is possible to couple the random variables $A$ and $\overline{A}$
such that they are equal with probability $1 - o(1)$. 
\label{lem:coupling}
\end{lemma}
\begin{proof}
Recall that, by assumption in Proposition \ref{prop:crossinglower}, all of the operations in $ S $ are on different keys. We can think of each insertion/deletion as
being performed on a random hash in $[n]$ (rather than being performed on any specific key). As part of our construction of $\overline{A}$,
we will end up modifying $S$ (i.e., adding and removing some operations) to get a new operation sequence $\overline{S}$. When adding new operations
to $S$, we will need not bother associating the new operations with actual keys, and will instead think of the new operations is simply
each being associated with a random hash. 

We now describe our construction of $\overline{A}$. Let $T_1, T_2, \ldots, T_{\sqrt{\mu}}$ be the time windows that correspond to the rows of $ G' $.
 Let $y_1, y_2, \ldots, y_{\sqrt{\mu}}, z_1 , z_2, \ldots, z_{\sqrt{\mu}}$ be independent Poisson random variables
with mean $|S| / (2\sqrt{\mu})$. Define $\overline{S}$ to be $S$, except that the operations in each time window $T_i$
are modified so that the number of insertions is $y_i$ and the number of deletions is $z_i$; note that this may require us to either add or remove
operations to the time window. Then define $\overline{A}$ in exactly the same way as
$A$, except using $\overline{S}$ in place of $ S $. That is, we let $\overline{a}_{i, j}$ be the number of insertions in $\overline{S}$ that pertain to cell $(i, j)$ in $G'$ and we let $\overline{b}_{i, j}$
the number of deletions in $\overline{S}$ that pertain to cell $(i, j)$ in $G'$. 

To understand the distribution of $\overline{A}$, we make use of an important fact about Poisson random variables: 
if $J$ balls are placed at random into $K$ bins, and $J$ is a Poisson random variable with mean $\phi$, then for each bin the number of balls in the bin is an independent Poisson random variable with mean $\phi / K$ (see Chapter 5.4 of \cite{MitzenmacherUp17}).
This implies that the $\overline{a}_{i, j}$'s and $\overline{b}_{i, j}$'s are independent Poisson random variables each of which has mean
$ 1$. 

To complete the proof, we must establish that $\Pr[A \neq \overline{A}] = o(1)$. For each time window $ T_i $, the expected number of operations that are in one of $ S $ or $\overline { S } $ but not the other
is 
$$O(\sqrt{|T_i|}) = O\left(\sqrt{|S| / \sqrt{\mu}}\right) = O\left(\sqrt{n / \sqrt{\mu}}\right) = O(\sqrt{n}).$$
In total over all $\sqrt{\mu}$ time windows $T_i$, the expected number of operations that are in one of $ S $ or $\overline { S } $ but not the other
 is therefore $O(\sqrt{n \mu})$.
The probability that any of these operations are on keys that hash to $ P $ is
$$O\left(\frac{|P|}{n} \cdot \sqrt{n \mu}\right) = O(|P|\sqrt{\mu / n}),$$
which by the assumptions of Proposition \ref{prop:crossinglower} is $o(1)$.
\end{proof}

Throughout the rest of the proof, we will treat the $a_{i, j}$s and $b_{i, j}$s as independent Poisson random variables each of which has mean $ 1$.  Our proof will make use of the fact that for any Poisson random variable $J$ with mean $\phi \ge 1$, 
\begin{align}
\Pr[J > \phi + \sqrt{\phi} t] \ge \exp(-\Omega(t^2)).
\label{eq:p1}
\end{align} 
For a derivation of \eqref{eq:p1}, see Theorem 1.2 of \cite{Pelekis16}.
Set $D = \sqrt{\log \log \mu} / c$, where $c$ is a sufficiently large constant. We will be making use of \eqref{eq:p1} in the case where $t = 2D$, that is, 
\begin{align}
\Pr[J > \phi + 2 \sqrt{\phi} D] & \ge \exp(-\Omega((\sqrt{\log \log \mu} / c)^2)) \nonumber \\
& = \exp(-\Omega(\log \log \mu / c^2)) \nonumber \\
& \ge 1 / \sqrt{\log \mu}.\label{eq:poissondeviate}
\end{align}

We now construct a monotone path $\gamma $ through the grid $ G' $, and show 
that $\gamma $ has blue-red deviation $\Omega(\sqrt{\mu \log \log \mu})$ with probability $1 - o(1)$.
The construction of $\gamma $ is recursive, with different levels of recursion operating on squares grids of different sizes.
To avoid ambiguity, we will use $k$ to refer to the height (or width) of the grid in the current recursive sub-problem, and we will use $\sqrt{\mu}$
to refer to the height (or width) of the grid in the top-level sub-problem.

The construction of $\gamma $ in a $k \times k$ subproblem is performed as follows. 
Break the $k \times k$ grid into four $k /2 \times k / 2$ quadrants. If the top left quadrant
has blue-red deviation at least $k D$ (that is, it contains at least  $k D$ more blue dots than red dots)
then we say that the subproblem \defn{successfully terminates}, and we 
set $\gamma$ to be the path that consists of $k$ vertical steps followed by $k$ horizontal steps.
Otherwise, we construct $\gamma$ by concatenating together a path recursively constructed through the bottom left quadrant and a path
recursively constructed through the top right quadrant. If a recursive subproblem is on a $1 \times 1$ grid, then we return the path
consisting of a vertical step followed by a horizontal step, and we call the subproblem a \defn{failed leaf}.

\begin{lemma}
Each subproblem with $ k >1 $ has probability at least $\Omega(1 / \sqrt{\log \mu})$ of successfully terminating.
\label{lem:term}
\end{lemma}
\begin{proof}
The number of blue dots and the number of red dots in the top left quadrant of the subproblem are both Poisson random variables
with means $k^2 / 4$. It follows by \eqref{eq:poissondeviate} that the number of blue dots has probability at least
$1 / \sqrt{\log \mu}$ of exceeding its mean by $k D$. Since the number of red dots has probability at least $\Omega(1)$
of being less than or equal to its mean, it follows that the blue-red deviation of the quadrant is at least $kD$
with probability at least $\Omega(1 / \sqrt{\log \mu})$.
\end{proof}

\begin{lemma}
With probability $ 1 - o (1) $, the number of failed leaves is less than $\sqrt{\mu}/2 $.
\label{lem:failedleaves}
\end{lemma}
\begin{proof}
We will prove that the expected number of failed leaves is $ o (\sqrt{\mu}) $, after which the lemma follows by Markov's inequality.

There are $\sqrt{\mu}$ potential failed leaves in the recursion tree, so it suffices to show that each of them has $ o (1) $ probability of occurring.
In order for a given failed leaf to occur, the recursion path of depth of $\log \sqrt{\mu}$ that must occur.
By Lemma \ref{lem:term}, each of the subproblems in the recursion path (except for the leaf) independently has at least a $\Omega(1 / \sqrt{\log \mu})$
probability of successfully terminating.
Thus each potential failed leaf in the recursion tree has probability at most
$$\left( 1 - \Omega(1 / \sqrt{\log \mu})\right)^{\log \sqrt{\mu} - 1} \le o(1)$$
of occurring.
\end{proof}

We can now analyze the blue-red deviation of $\gamma $ to prove Proposition \ref{prop:surpluslower}.
\begin{proof}[Proof of Proposition \ref{prop:surpluslower}]
By Lemma \ref{lem:failedleaves}, we may assume that the number of failed leaves is less than $\sqrt{\mu} / 2$.

Say that the \defn{width} of a subproblem is width of the grid in which it takes place. 
The sum of the widths of the leaf subproblems is $\sqrt{\mu}$. Each failed leaf has width 1, so if
the number of failed leaves is less than $\sqrt{\mu}/2 $, then the sum of the widths of the 
leaves that successfully terminate must be at least $\sqrt{\mu} / 2$. 

For each leaf subproblem with width $k$ that successfully terminates, its top left quadrant
contributes $\Omega(k \sqrt{\log \log \mu})$ to the blue-red deviation of $\gamma $.
Summing over the leaf subproblems that successfully terminate, the top left quadrants of all of them
contribute at least $\Omega(\sqrt{\mu \log \log \mu})$ to the blue-red deviation.

Additionally, we must consider the effect of the blue and red dots below $\gamma $ that are
not contained in any of these aforementioned top-left quadrants. 
Once the path $\gamma $ is determined, the number $X$ of such blue dots
and the number $Y$ of such red dots are independent Poisson random variables
satisfying $\E[X] = \E[Y] \le \mu$. By a Chernoff bound, we have that with probability $1 - o(1)$,
$$Y - X \le \sqrt{\mu \log \log \log \mu}.$$
Thus, with probability $ 1 - o (1) $, the blue-red deviation of $\gamma $ is at least
$$\Omega(\sqrt{\mu \log \log \mu}) -  \sqrt{\mu \log \log \log \mu} = \Omega(\sqrt{\mu \log \log \mu}).$$
\end{proof}

\section{Relating Insertion Surplus to Crossing Numbers}\label{sec:cross}

In this section we use our insertion-surplus bounds from Section \ref{sec:surplus} to 
obtain bounds on a different set of quantities $\{c_j\}_{j \in [n]}$ that we call the crossing numbers; later, in 
Section \ref{sec:time}, our bounds on crossing numbers will allow for us to analyze
the amortized costs of insertions/deletions/queries in ordered linear probing.

Consider an ordered linear probing hash table with $ n $ slots, and 
suppose that the hash table is initialized to have load factor $ 1 - 1/x $ or smaller.
Consider a sequence of insertion and deletion operations $ S $ such that the 
load factor never exceeds $1 - 1/x$. (Note that, unlike in Section \ref{sec:surplus}, the lemmas in this section will
not all require that $S$ alternates between insertions and deletions.)

Based on the initial state of the hash table and on the sequence $ S $ of operations,
define the \defn{crossing numbers} $c_1, c_2, \ldots, c_n$ so that $c_i$ is the number of times
that an insertion with a hash smaller than $ i $ either (a) uses a tombstone left
by a key that had hash at least $ i $; or (b) uses a free slot in a position
greater than or equal to $ i $.

The purpose of this section is to prove nearly tight bounds on $\E[c_i]$. 
Subsequent sections will then show how to use these bounds in order to analyze
the performance of linear probing.

We will need the following additional definitions. Define the \defn{insertion surplus} of a subinterval $P \subseteq [n]$ to be the
maximum insertion surplus of any downward-closed subset of $\{u \in S \mid h(u) \in P\}$,
minus the number of free slots that are initially in the range $ P $. 
Define the \defn{peak} $p_u$ of an insertion $u$ to be the hash of the tombstone that
the insertion uses (if it uses a tombstone)
or the position of the free slot that the insertion uses (if it uses a free slot).

The following lemmas characterize the crossing numbers in terms of the insertion surpluses of intervals.

\begin{lemma}
For each $s \in [n]$, there exists an interval $P = [r, s - 1]$ whose insertion surplus is at least $ c_s $.
\label{lem:crossingsurplus}
\end{lemma}
\begin{proof}
Call an insertion $u$ with hash smaller than $s$ \defn{special} if it satisfies the following recursive property: 
either (a) $p_u \ge s$; or (b) there is another special insertion $v$ that occurs temporally after $u$ such 
that $p_u \ge h(v)$. Call a deletion $ v $ \defn{special} if $h(v) < s$ and there exists a special insertion $u$ that occurs temporally after
$ v $ and satisfies $ h (u) \le h (v) $.

Let $r$ be the smallest hash of any special insertion/deletion, and define $ P = [r, s -1] $. 
We will prove that the insertion surplus of $ P $ is at least $ c_s $. Towards this end,
define $ A $ to be the number of special insertions, define $ B $ to be the number of special
deletions, and define $ C $ to be the number of free slots initially in $ P $. The set of special operations
is downward-closed by design, and its insertion surplus is $A - B - C$. Thus we wish to show that
$$A - B - C \ge c_s.$$

In order for an insertion $u$ to contribute to the crossing number $ c_s $, the insertion must have peak $p_u \ge s$ and thus must also be special. To complete the proof, we will show that there are at least $ B + C $ special insertions with peaks $ p_u <s $ (and thus at most $A - B - C$ special insertions have $p_u \ge s$). That is, we will show that every
tombstone created by a special deletion and every free slot initially in $ P $ is used by some special insertion.

Consider a tombstone that is created by a special deletion $v$. Since $ v $ is special, there must exist a special insertion $ u $ that occurs after $ v $
and satisfies $ h (u) \le h (v) $. Let $u$ be the last such insertion. We must have that $ p_u  > h(v) $, since otherwise, we could arrive at a contradiction as follows: in order so that $ u $ could be special,
there would have to be a subsequent special insertion $ z $ with $h(z) \le p_u$; but this would imply that $h(z) \le h(v)$, which would contradict the fact that $ u $ is 
the final special insertion satisfying $ h (u) \le h (v) $. Since $ p_u  > h(v) $, it must be that the tombstone created by the deletion $ v $ has already been used by the time that
insertion $ u $ is performed. The insertion $w$ that used the tombstone must have occurred before the insertion $ u $ and must have had peak
$p_w = h(v) \ge h(u)$. This means that $ w $ is itself a special insertion. Thus the tombstone created by $ v $ is used by a special insertion, as desired. 

Now consider a free slot $j$ that is initially present in $ P $. By the definition of $ P $, there must exist a special insertion $ u $ that satisfies $ h (u) \le j $. 
Let $u$ be the last such insertion. We must have that $ p_u  > j$, since otherwise, we could arrive at a contradiction as follows: in order so that $ u $ could be special,
there would have to be a subsequent special insertion $ z $ with $h(z) \le p_u$; but this would imply that $h(z) \le j$, which would contradict the fact that $ u $ is 
the final special insertion satisfying $ h (u) \le j$. Since $ p_u  > j$, it must be that the free slot $j$ has already been used by the time that
insertion $ u $ is performed. The insertion $w$ that used slot $j$ must have occurred before the insertion $ u $ and must have had peak
$p_w = j \ge h(u)$. This means that $ w $ is itself a special insertion. Thus the free slot is used by a special insertion, as desired. 
\end{proof}

The converse of the previous lemma is also true.
\begin{lemma}
If there exists an interval $ P = [r, s -1] $ with insertion surplus $ q $, then $ c_s\ge q $.
\label{lem:crossingsurpluslower}
\end{lemma}
\begin{proof}
Let $ S' $ be the downward-closed subset of $\{u \in S \mid h(u) \in P\}$ with the largest insertion surplus.
Every insertion in $ S' $ must either (a) use a tombstone created by a deletion in $ S' $, (b) use a free slot
initially present in $ P $, or (c) have peak at least $ s $. 
It follows that if $ A $ is the number of insertions in $ S' $, $ B $ is the number of deletions in $ S' $, and $ C $ is the number of free slots initially
in $ P $, then we must have that
$$ c_s \ge A - B - C.$$
Since the quantity on the right side is exactly the insertion surplus of $ P $, the proof is complete.
\end{proof}

The previous lemmas tells us that, in order to understand the crossing numbers $ c_s $, it suffices to
understand the insertion surplus of each interval $ P $. This insertion surplus, in turn, depends on two quantities:
the maximum insertion surplus of any downward-closed subset of $\{u \in S \mid h(u) \in P\}$; and the number of free slots
initially in $ P $. We have already achieved a good understanding of the first quantity in the previous section.
The next two lemmas analyze the second quantity.

\begin{lemma}
Suppose that the hash table initially has load factor $ 1 -\epsilon $. Consider any interval $P \subseteq[n]$ of size at least
$c \, \epsilon^{-2}\log \epsilon^{-1}$, where $ c $ is taken to be a sufficiently large constant. With probability $ 1 -1 / \poly |P|$,
the interval $ P $ initially contains at least $\Omega(\epsilon |P|)$ free slots. 
\label{lem:manyfreeslots}
\end{lemma}
\begin{proof}
The expected number of elements that hash into $ P $ initially is $(1 - \epsilon)|P|$, which since $|P| \ge c \, \epsilon^{-2}\log \epsilon^{-1}$, is at most $|P| - \sqrt{|P| c \log |P|}$. 
It follows by a Chernoff bound that, with probability $1 - 1/ \poly(|P|)$, the number of elements that initially hash into $ P $ is
at most  $|P| - \frac{1}{2} \sqrt{|P| c \log |P|}$. On the other hand, by Corollary \ref{cor:over}, and with probability $1 - 1/ \poly(|P|)$,
the number of elements that reside in $ P $ but hash to a position prior to $ P $ is at most $O(\epsilon^{-1} \log P) \le \frac{1}{4} \sqrt{|P| c \log |P|}$.
The total number of elements that reside in $ P $ is therefore at most $|P| - \frac{1}{4} \sqrt{|P| c \log |P|}$, which completes the proof.
\end{proof}

\begin{lemma}
Suppose that the hash table initially has load factor $ 1 -\epsilon $. Consider any interval $P = [a, b] \subseteq[n]$ of size 
$\epsilon^{-2} / c$, where $ c $ is taken to be a sufficiently large constant. With probability $\Omega (1) $,
there are initially no free slots in $ P $.
\label{lem:nofreeslots}
\end{lemma}
\begin{proof}
Consider the state of the hash table initially, and let $ r $ be the length of the run of non-free slots beginning at position $a$.
Knuth in \cite{Knuth63} established that $\E[r] = \Theta(\epsilon^{-2})$. On the other hand, Proposition \ref{prop:classic} tells us that
$\Pr[r > k\epsilon^{-2}] \le 1 / \poly(k)$ for all $ k $. The only way that these two facts can be consistent is if
$r = \Omega(\epsilon^{-2})$ with probability $\Omega(1)$. Thus the lemma is established. 
\end{proof}

We are now in a position to upper bound the crossing number $ c_j $. 
\begin{proposition}
Suppose that the hash table initially has load factor $ 1 - 1/x $, 
suppose that $|S| = \Omega(n / x)$ and $|S| \le n$, and suppose that $ S $ alternates between insertions and deletions.

There exists a positive constant $d$ such that for any $j \in [n]$ and any 
$$r \ge \sqrt{\frac{|S|}{n} x^{2}} \log^d x,$$
we have $q_j < r$ with probability $1 - 1 / \poly(r)$.
\label{prop:crossingupper}
\end{proposition}
\begin{proof}
Define 
$$\mu_i = |S| \cdot i / n$$ to be the expected number of operations in $S$ that hash into the $[j - i, j - 1]$.
Define $$\lambda_i = \sqrt{\mu_i} \polylog \mu_i$$ where the polylogarithmic factor is selected so that Proposition \ref{prop:surplus} 
offers the following guarantee: with probability $1 - 1 / \poly(\mu_i)$, every downward-closed subset of
 $\{u \in S \mid h(u) \in [j - i, j - 1]\}$ has insertion surplus less than $\lambda_i$.
Let $$K = x^{2} \log^c x$$ for some sufficiently large positive constant $ c $.
This results in
\begin{align*}
\lambda_K & = \sqrt{\frac{|S|}{n} x^{2} \log^c x} \polylog \left(\frac{|S|}{n} x^{2} \log^c x\right) \\
          & = \sqrt{\frac{|S|}{n} x^{2}} \polylog x.
\end{align*}
Thus, if we select the constant $d$ in the proposition statement appropriately, then the requirement that $ r \ge \sqrt{\frac{|S|}{n} x^{2}} \log^d x$ implies that $r \ge \lambda_K$, and thus that $r = \lambda_R$ for some $R \ge K$. To prove the proposition, it suffices to establish that for every $ R \ge K$, we have
\begin{equation}
\Pr[c_j \ge \lambda_R] \le 1 / \poly(\lambda_R).
\label{eq:lambda}
\end{equation}
Note that in the parameter regime $R \ge K $, we have that $\poly (\lambda_R) = \poly(\mu_R) = \poly (R) $ (here we are using that $\Omega(n / x) \le |S| \le O(n)$), so we will treat the three as interchangeable throughout the rest of the proof.

Define $P_i = [j - i, j - 1]$ and define $s(P_i)$ to be the insertion surplus of $P_i$. By Lemma \ref{lem:crossingsurplus},
\begin{align*}
\Pr[c_j \ge \lambda_R] & \le \Pr[s(P_i) \ge \lambda_R  \text{ for some } i] \\
                       & \le \Pr[s(P_i) \ge \lambda_R  \text{ for some } i < R] +  \sum_{i \ge R} \Pr[s(P_i) > 0].
\end{align*}

To prove \eqref{eq:lambda}, we begin by bounding $\Pr[s(P_i) \ge \lambda_R  \text{ for some } i < R]$. 
If $s(P_i) \ge \lambda_R$ for some $ i < R $, then there must be a downward-closed subset $ S' $ of
 $\{u \in S \mid h(u) \in [j - (R - 1), j - 1]\}$ such that the insertion surplus of $ S' $ is at least $\lambda_R $. But by Proposition \ref{prop:surplus} 
and by the definition of $\lambda_{ R -1 } $, we know that with probability $1 - 1 / \poly(\mu_{R - 1})$ (and thus also in $\lambda_{R - 1}$), every such $ S' $ has insertion surplus at most
$\lambda_{ R -1 } <\lambda_R $. Thus the probability that $s(P_i) \ge \lambda_R$ for any $ i < R $ is at most $ 1/\poly (\lambda_R)$. 

To complete the proof, it remains to show that
$$\sum_{i \ge R} \Pr[s(P_i) > 0] \le \frac{1}{\poly (\lambda_R)}.$$
We will establish a stronger statement, namely that for every $i \ge K$,
$$\Pr[s(P_i) > 0] \le \frac{1}{\poly(i)}.$$

Since $i > K$, we can apply Lemma \ref{lem:manyfreeslots} to deduce that, with probability $1 - 1 / \poly(i)$, the interval $P_i$ initially contains at least 
$
\Omega(i / x)
$
free slots. We further have that, by Proposition \ref{prop:surplus}, and with probability $ 1-1/\poly (\mu_i) = 1-1/\poly (i)$, 
every downward-closed subset of
 $\{u \in S \mid h(u) \in [j - i, j - 1]\}$ has insertion surplus less than $\lambda_i$. 
It follows that $ s (P_i) $ is that most
$$\max(0, \lambda_i -\Omega(i / x)). $$
In order to establish that $ s (P_i) $ is zero, it suffices to show that
$$\lambda_i = o(i / x).$$
This is simply a matter of calculation:
\begin{align*}
\lambda_i & = \sqrt{\mu_i} \polylog \mu_i  & \text{ (by definition of $\lambda_i$)}\\
          & = \sqrt{\frac{|S| i }{n}} \polylog \left(\frac{|S| i }{n} \right) &\text{(by definition of $\mu_i$)}\\
          & \le  \sqrt{i} \polylog i & \text{(since $|S| \le n$)}\\
          & \le  O\left(\frac{i \polylog i}{\sqrt{K \lg^c (i / K)}}\right) & \text{(since $i > K$)}\\
          &  = O\left(\frac{i \polylog i}{\sqrt{x^2 \log^c x \lg^c (i / K)}}\right) & \text{(since $K = x^2 \lg^c x$)}\\
          &  = O\left(\frac{i \polylog i}{\sqrt{x^2 \log^c K \lg^c (i / K)}}\right) & \text{(since $K = x^2 \lg^c x$)}\\
          &  = O\left(\frac{i \polylog i}{\sqrt{x^2 \log^c i}}\right) & \text{(since $\lg a \lg b \ge \Omega(\lg (ab))$)}\\
          & = \frac{i / x \polylog i}{\sqrt{\log^c i}} &\\
          & = o(i / x) & \text{(since $c$ is a sufficiently large constant).} 
\end{align*}
\end{proof}

\begin{corollary}
Suppose that the hash table initially has load factor $ 1 - 1/x $, 
suppose that $|S| = \Omega(n / x)$ and $|S| \le n$, and suppose that $ S $ alternates between insertions and deletions.

For each $j \in [n]$, 
$$\E[c_j] \le \sqrt{\frac{|S|}{n} x^{2}} \polylog x.$$
\label{cor:crossingupper}
\end{corollary}

We can also obtain a nearly matching lower bound for $\E [c_j] $.
\begin{proposition}
Suppose that the hash table initially has load factor $ 1 - 1/x $,
suppose that $|S| = \Omega(n / x)$ and $|S| \le n$, and suppose that $ S $ alternates between insertions and deletions.

Further suppose that each operation in $ S $ applies to a different key.
Then for each $ j\in [n] $,
$$\E[c_j] = \Omega\left(\sqrt{\frac{|S|}{n} x^{2}\log\log x^{} }\right) .$$
\label{prop:crossinglower}
\end{proposition}
\begin{proof}
By Lemma  \ref{lem:crossingsurpluslower}, it suffices to show that there is some interval $ P_i = [j - i, j -1] $ with insertion surplus
$$\Omega\left(\sqrt{\frac{|S|}{n} x^{2} \log\log x^{} }\right).$$
By Lemma \ref{lem:nofreeslots}, there exists a positive constant $ c $ such that, with probability $\Omega (1) $, the interval $ P_{ cx^{2 } } $
in the hash table initially (i.e., at the beginning of the operations $S$) contains no free slots. 
Furthermore, Proposition \ref{prop:surpluslower} tells us that with probability $ 1 - o (1) $,
there exists a downward-closed subset of $\{u \in S \mid h(u) \in P_{cx^2}\}$ with insertion surplus at least
\begin{align*}
& \Omega\left(\sqrt{\frac{|S|}{n} x^{2} \log \log \left(\frac{|S|}{n} x^{2}\right)}\right) \\
& = \Omega\left(\sqrt{\frac{|S|}{n} x^{2} \log\log x^{} }\right).
\end{align*}
With probability $\Omega (1) $, both of the preceding events occur simultaneously. Thus the proposition is proven.
\end{proof}

The previous two propositions both focus on the case in which the workload $ S $ alternates between insertions and deletions.
We conclude this section by considering the case where $ S $ is allowed to perform an arbitrary sequence of insertions and deletions,
subject only to the constraint that the load factor never exceeds $ 1 - 1/x $. 

\begin{proposition}
Suppose that the hash table begins at a load factor of at most $ 1 - 1/x $, and that the sequence of operations $ S $
keeps the load factor at or below $ 1 - 1/x $. Finally, suppose that $|S| \le n / \polylog(x)$. Then for each $ j\in [n] $,
$$\E[ c_j ] = O (x^{}). $$
\label{prop:amortized}
\end{proposition}
\begin{proof}
We begin by constructing an alternative sequence of insertions/deletions $\overline { S } $ such that
the crossing numbers with respect to $\overline { S } $ are guaranteed to be at least as large as the crossing numbers with respect to $ S $.
We will then complete the proof by analyzing the crossing numbers of $\overline { S } $.

Suppose that the hash table initially contains $ (1 - 1/x) n - w $ elements for some $ w $.
We construct $\overline { S } $ through two steps: 
\begin{itemize}
\item Call an insertion \defn{novel} if the key being inserted has not been inserted in the past and was not originally present in the hash table.
The first step is to take each of the first $ w $ novel insertions, and to move them to the front of the operation sequence.\footnote{We can assume without loss of generality that there are at least $ w $ such novel insertions, 
since if there are not, we can artificially add additional insertions to the end of $ S $ and then perform the rest of the proof without modification.}
\item Call a triple of three consecutive operations \defn{unbalanced} if the triple consists of two deletions followed by an insertion.
At least one of the two deletions in such a triple must act on a different key than the insertion acts on. We can \defn{balance} 
the triple by changing the order of operations so that the aforementioned deletion occurs last. The second step in constructing $\overline { S } $
is to repeatedly find and balance any unbalanced triples until there are no such triples left.
\end{itemize}

Observe that the sequence $\overline { S } $ is a valid sequence of operations, since the transformation from $ S $ to $\overline { S } $
never swaps the order of any two operations that act on the same key.

We claim that the crossing numbers with respect to $\overline { S } $ are at least as large as the crossing numbers with respect to $ S $.
The transformation from $ S $ to $\overline { S } $ moves certain insertions to occur earlier than they would have otherwise, and certain deletions to occur
later than they would have otherwise. Importantly, these types of moves cannot decrease the insertion surplus of any interval $ P $ in the hash table.
By Lemmas~\ref{lem:crossingsurplus} and \ref{lem:crossingsurpluslower}, the crossing numbers are completely determined by the insertion surpluses of the intervals $ P \subseteq [n]$.
Since the insertion surpluses for $\overline { S } $ are at least as large as those for $ S $, it follows that the crossing numbers for $\overline { S } $ are also
at least as large as those for $ S $.

Our next claim is that $\overline { S } $ never causes the load factor to exceed $ 1 - 1/x  $.
This can be seen by analyzing each of the two steps of the construction separately. The first step 
modifies only the window of time in which the first $ w $ novel insertions are performed; no rearrangement of the
operations in this window of time can possibly cause the load factor to exceed $ 1 - 1/x  $.
The second step repeatedly performs balancing operations on unbalanced triples; such a balancing operation 
does not change the maximum load factor that is achieved during the triple, however, since that load factor is achieved
prior to the first operation of the triple. Combining the analyses of the two steps, we see that the load factor never exceeds $ 1 - 1/x  $.

Now let us reason about the structure of $\overline { S } $. By design, $\overline { S } $ begins with $ w $ insertions, bringing the load factor up to exactly $ 1 - 1/x $.
Since the load factor never exceeds $ 1 - 1/x $, and since there are never two deletions in a row followed by an insertion, it must be that the remaining insertions in $\overline { S } $
are each preceded by exactly one deletion. In other words, $\overline { S } $ must consist of three parts $\overline { S }_1, \overline { S }_2, \overline { S }_3$
where $\overline { S }_1$ consists only of insertions, $\overline { S }_2 $ alternates between deletions and insertions, and $\overline { S }_3 $ consists only of deletions.

Since $\overline { S }_3 $ consists only of deletions, it does not contribute anything to the crossing numbers.
By Corollary \ref{cor:crossingupper}, the expected contribution of the operations in $\overline { S }_2 $ to each crossing number $c_j$ is at most $O (x^{}) $.

It remains to bound the contribution of $\overline{S}_1$ to the crossing numbers. If an insertion takes time $t $,
then it can contribute at most $t$ to the sum $\sum_j c_j$. 
Knuth showed in \cite{Knuth63} that the total time needed to fill an empty table up to a load factor of $1 - 1/x $ 
is $O(n x) $ in expectation. Thus the expected contribution of $\overline{S}_1$ to $\sum_j c_j$
is $O(n x) $, completing the proof. 
\end{proof}

\section{Relating Crossing Numbers to Running Times}\label{sec:time}

In this section, we give nearly tight bounds on the performance of ordered linear probing. Notably, we find that,
if the size $R$ of each rebuild window is chosen correctly, then 
the amortized time per insertion is guaranteed to be $\tilde{O}(x)$. The key technical component
to the section will be a series of arguments transforming our bounds on crossing numbers (in Section \ref{sec:cross})
into bounds on running times.

Consider an ordered linear probing hash table that uses tombstones for deletions. 
Recall that there are three parameters: the number $ n $ of slots in the table, the number $ R $ of insertions in each time window
between rebuilds, and the maximum load factor $ 1 - 1/x $ that the hash table ever achieves. Based on these parameters, we 
wish to analyze the average running time of the operations being performed on the hash table.

We will be focusing exclusively on the regime in which $ R = \Omega (n / x) $ and $R \le n$. 
Since each rebuild can be implemented in linear time $ O (n) $, the average time spent performing rebuilds per operation is
$ O (n/R) = O (x^{})$ (which for our purposes will be negligible). 
Thus the focus of our analysis will be on analyzing the costs of the operations that occur between consecutive rebuilds.

Before diving into the details, we remark that there are two main technical challenges
that our analysis must overcome. The first challenge is obvious: we must quantify the degree to which
tombstones left behind by deletions improve the performance of subsequent insertions.
The second challenge is a bit more subtle: in order to support large rebuild-window sizes $ R $,
our analysis must be robust to the fact that tombstones can accumulate over time, increasing the effective
load factor of the hash table. 
This latter challenge is exacerbated by the fact that the choice of which tombstones are in the table
at any given moment is a function not only of the sequence of operations being performed, but also of the
randomness in the hash table. This means that, even if the cumulative load factor from the elements and tombstones can be bounded (e.g., by $ 1 -1 / (2x) $), we still cannot apply the classic analysis at that load factor
in order to bound the expected time of queries.

One of the interesting features of our analysis is that we completely circumvent the issue of how fast
tombstones accumulate over time. Rather than focusing on what the effective load factor of the hash table is at each moment in time,
the analysis instead analyzes the state of the hash table at the beginning of the rebuild window, and then analyzes the dynamics of how
the local structure of the hash table changes over time.

In the following lemmas, we will focus on a single window $ W $ of time between two rebuilds.
We begin by defining three quantities that we will be able to express the running times of operations in terms of.

For a given position $ i\in [n] $, define the \defn{positional offset $o_i$} 
to be the quantity $j - i$ where $j \ge i$ is the largest position such that,
at the end of the time window $ W $, all of the positions $[i, j -1]$
contain elements and tombstones whose hashes are smaller than $ i $. Note that,
although the positional offset is defined at the end of the time window $ W $,
if we were to define the same quantity at any other point during the time window,
it would be at most $ o_i $ (that is, the positional offset only increases over time).

For a given position $ i\in [n] $, define the \defn{spillover $ s_i $} to be the largest $k \ge 0$ 
such that if we consider all keys that are either initially present or inserted at some point during $ W $,
at least $ 4k $ of them have hashes in the range $ [i, i + k) $.

For each insertion $ u $, define the  \defn{displacement $ d_u $} of the insertion to be $ p_u - h (u) $, where $ p_u $ is the
peak of the insertion as defined in Section \ref{sec:cross}.

\begin{lemma}
If an insertion $ u $ hashes to a position $ i $, then the insertion takes time at most
$$O(o_i + s_i + d_u + 1).$$
Similarly, if a query/deletion $ u $ hashes to a position $ i $, then the operation takes time at most
$$O(o_i + s_i + 1). $$
\label{lem:runningtimedecomposition}
\end{lemma}
\begin{proof}
Consider an insertion $ u $ that hashes to a position $ i $. If $ u $ uses a free slot in some position $i + r$,
then the time to perform the insertion is $r = d_u + 1$. 
Suppose, on the other hand, that $ u $ uses a tombstone with some hash $ i + t $, and the tombstone is in some position $i + r$.
Then the running time of the insertion is $r$, and we wish to show that 
\begin{equation}
r = O(o_i + s_i + t + 1).
\label{eq:r}
\end{equation}
If $ r - o_i = O(t)$, then \eqref{eq:r} trivially holds and we are done. 
Otherwise, we may assume that $r - o_i \ge 4(t + 1)$.
By the definition of the positional offset $ o_i $, all of the elements/tombstones in positions $[i + o_i, i + r]$ must have 
hashes in the range $[i, i + t]$. Combining this with the fact that $r - o_i \ge 4(t + 1)$, 
it follows that the spillover $s_i$ satisfies $s_i \ge \lfloor (r - o_i) / 4 \rfloor$, hence \eqref{eq:r}.

Next consider a query/deletion $ u $ that hashes to a position $ i $.
The operation takes time at most $ O (o_i + T) $ where $ T $ is the total number of elements
with hash $ i $ that are either present at the beginning of $ W $ or inserted at some point during $ W $.
By the definition of $ s_i $, we have that $ T \le 4s_i + O (1) $.
Thus the operation takes time $O(o_i + s_i + 1). $

\end{proof}

Our next lemma relates $ o_i $, $ s_i $, and $ d_u $ to the crossing numbers $ c_i $ defined in the previous section.
\begin{lemma}
For each $ i\in [n] $, the positional offset $ o_i $ satisfies $o_i \ge c_i$ and
\begin{equation}
\E [o_i] = O(x^{}) +\E [c_i],
\label{eq:oi}
\end{equation}
and the spillover $ s_i $ satisfies
\begin{equation}
\E [s_i] = O(1).
\label{eq:si}
\end{equation}
Finally, if we consider a random insertion $u$ in the time window $ W $, then
\begin{equation}
\E [d_u] = \frac{n}{R} \E [c_i].
\label{eq:dx}
\end{equation}
\label{lem:positionaloffset}
\end{lemma}
\begin{proof}
Although the positional offset $ o_i $ is only defined at the end of the time window $ W $, let us slightly abuse notation
and think about how the quantity evolves over time (that is, what would happen if we defined the quantity at each point in time in the time window).
By Corollary \ref{cor:over}, the initial positional offset (at the beginning of $W$) has expected value $ O (x^{}) $.
During the time window, the positional offset increases by one each time that an insertion whose hash is less than $ i $
has a peak that is at least $ i $. The number of such insertions is precisely $ c_i $. Since these insertions are the only operations that can change
the positional offset, it follows that  $o_i \ge c_i$ and
$\E [o_i] = O(x^{}) +\E [c_i]$.

Next we consider the spillover $ s_i $. Let $ V $ be the set of all elements that are present at some point during $ W $.
Then $|V| \le 2n$, and the expected number of elements in $V$ that hash to a given position $ j $ is at most $ 2 $.
It follows by Lemma \ref{lem:saturated} that $\Pr[s_i > k] \le \exp(-\Omega(k))$ for all $ k $. This implies
\eqref{eq:si}.

Finally we establish \eqref{eq:dx}. 
Observe that, if an insertion $ u $ has displacement $ d_u $, then the insertion contributes to exactly
$ d_u $ crossing numbers $ c_i $. It follows that
$$\sum_{u } d_u =\sum_{i\in [n] } c_i. $$
If we select a random insertion $ u $ out of the $ R $ insertions that occur in $ W $, then
$$R \cdot \E [d_u] = \sum_{i\in [n] } c_i.$$
Since the $ c_i $'s all of the same expected values, it follows that for a given $ i $,
$$R \cdot \E [d_u] = n  \cdot \E[c_i].$$
This implies \eqref{eq:dx}.
\end{proof}

We are now prepared to prove the main results of the section.
We begin by considering a hovering workload, that is a workload in which queries can be performed at arbitrary times, but insertions and deletions must alternate.
\begin{theorem}
Consider an ordered linear probing hash table that uses tombstones for deletions, and that performs rebuilds every $ R $ insertions.
Suppose that the table is initialized to have capacity $ n $ and load factor $ 1 - 1/x $, where $R = \Omega(n / x)$ and $R \le n$.
 Finally, consider a sequence $S$ of operations that alternates between insertions and deletions (and contains arbitrarily many queries).

Then the expected amortized time $I$ spent per insertion satisfies
\begin{equation}
I \le \tilde{O}\left(x\sqrt{\frac{n}{R}}\right)
\label{eq:f1}
\end{equation}
and, if all insertions/deletions in each rebuild window are on distinct keys, then
\begin{equation}
I \ge \Omega\left(x\sqrt{\frac{n}{R} \log\log x}\right).
\label{eq:f2}
\end{equation}
Moreover, the expected time $Q$ of a given query/deletion satisfies
\begin{equation}
Q \le O(x) + \tilde{O}\left(x\sqrt{\frac{R}{n}}\right)
\label{eq:f3}
\end{equation}
and, if all operations in each rebuild window are on distinct keys, then for any negative query at the end of a rebuild window, we have
\begin{equation}
Q \ge \Omega\left(x +  x \sqrt{\frac{R}{n}\log\log x} \right).
\label{eq:f4}
\end{equation}
\label{thm:hovering}
\end{theorem}
\begin{proof}
Consider a random insertion $ u $ with some hash $ i $. By Lemma \ref{lem:runningtimedecomposition}, we have that $ u $ takes time at most
$O(o_i + s_i + d_u + 1).$ By Lemma \ref{lem:positionaloffset}, this has expectation at most
$$O\left(x + \E[c_i] + \frac{n}{R} \E[c_i] + 1\right) = O\left(x + \frac{n}{R} \E[c_i]\right).$$
By Corollary \ref{cor:crossingupper}, this is that most\footnote{There is one technicality that we must be slightly careful about here: the hash $i = h(u)$ is independent
of where every key $v \neq u$ hashes to, but is (trivially) not independent of where key $u$ hashes to. However, since $u$ is only one key, it can easily be factored out of the analysis so that we can treat $i$ as being a random slot (independent of the hash function $h$).}
$$\tilde{O}\left(x + \frac{n}{R} \sqrt{\frac{R}{n} x^{2}}\right) =  \tilde{O}\left(x\sqrt{\frac{n}{R}}\right).$$

On the other hand, the insertion time is necessarily at least $d_u$, which by Lemma \ref{lem:positionaloffset}, has expectation
$$\frac{n}{R - 1} \E[c_j]$$
for each $j \in [n]$.
If we assume that every insertion/deletion in the rebuild window is on a different key, then we can further apply Proposition \ref{prop:crossinglower} to conclude that the insertion time has expected value at least
$$\Omega\left(\frac{n}{R} \sqrt{\frac{R}{n} x^{2} \log \log x}\right) =  \Omega\left(x\sqrt{\frac{n}{R} \log \log x}\right).$$

Now instead consider a query/deletion $u$ that hashes to some position $ i $.  By Lemma \ref{lem:runningtimedecomposition}, we have that $ u $ takes time at most $O(o_i + s_i  + 1).$ By Lemma \ref{lem:positionaloffset}, the expected time that $ u $ takes is therefore at most 
$$O\left(x + \E[c_i] + 1\right).$$
By Corollary \ref{cor:crossingupper}, this is that most
$$\tilde{O}\left(x + \sqrt{\frac{R}{n} x^{2}}\right) =  \tilde{O}\left(x + x\sqrt{\frac{R}{n}}\right).$$

If we assume that the query is a negative query, performed at the end of a rebuild window whose insert/delete operations are all on different keys, then the query time is necessarily at least $o_i$,
which by Lemma \ref{lem:positionaloffset} is at least $c_i$. It follows by Proposition \ref{prop:crossinglower} that the expected query time is at least
$$\Omega\left(x \sqrt{\frac{R}{n}\log\log x} \right).$$
The expected query time is also $\Omega(x)$ trivially, by the standard analysis of linear probing \cite{Knuth63}.

\end{proof}

Theorem \ref{thm:hovering} has several important corollaries. Our first corollary considers the
setting in which rebuilds are performed every $\Theta(n / x)$ insertions.
\begin{corollary}
Suppose $R = \Theta(n / x)$. Then 
$$I \le \tilde{O}(x^{1.5}),$$
and, if all insertions/deletions in each rebuild window are on distinct keys, then 
$$I \ge \Omega(x^{1.5} \sqrt{\log \log x}).$$
Moreover, $Q = \Theta(x)$.
\label{cor:smallrebuild}
\end{corollary}

Our next corollary considers the setting in which rebuilds are performed every $n / \polylog(x)$ insertions.
In this case, the hash table achieves nearly optimal scaling.
\begin{corollary}
If $R = n / \log^c x$ for a sufficiently large positive constant $ c $, then 
$$I = \tilde{\Theta}(x).$$
Moreover, $Q = \Theta(x)$.
\label{cor:largerebuild}
\end{corollary}

Our final corollary considers the question of whether it is possible to select a value of $R$ that allows for both $I$ and $Q$ to be $O(x)$. 
The corollary establishes that no such $ R $ exists.
\begin{corollary}
For every choice of $ R $,  there exists $S$ such that either $I = \omega(x)$ or $Q = \omega(x)$. 
\end{corollary}
\begin{proof}
If $n / R = \omega(1)$, then \eqref{eq:f2} tells us that there exists a sequence of operations for which $I$ is $\omega(x)$. On the other hand, if $n / R = o(\lg \lg x)$, then \eqref{eq:f4} tells us that there exists a sequence of operations for which $Q$ is $\omega(x)$. 
\end{proof}

Up until now, we have been focusing on a hovering workload. Our final result considers an arbitrary workload of operations,
where the only constraint is that the load factor never exceeds $ 1 - 1/x $.
Notice that if $ R $ is small (i.e., $R = \Theta(n / x)$), then ordered linear probing can potentially perform very poorly in the setting,
since an entire rebuild window can potentially consist of only insertions, none of which are able to make use of tombstones, but all
of which are performed at a load factor of $1 - \Theta(1 / x)$.
Our next theorem establishes, however, that if $ R $ is selected appropriately, then the amortized performance of
ordered linear probing is near the optimal $O(x)$ that one could hope to achieve.

\begin{theorem}
Let $ c $ be a sufficiently large positive constant. 
Consider an ordered linear probing hash table that uses tombstones for deletions, and that performs rebuilds every $ R = n / \log^{c} x$ insertions.
Finally, consider a sequence of operations $ S $ that never brings the load factor above $ 1 - 1/x $.

Then the expected amortized cost of each insertion is $\tilde{O}(x)$ and the expected cost of each query/deletion is $O(x)$.
\label{thm:amortized}
\end{theorem}
\begin{proof}
This follows from the same proof as Theorem~\ref{thm:hovering}, except using
Proposition \ref{prop:amortized} instead of Corollary~\ref{cor:crossingupper}.
\end{proof}

\begin{remark}
The proofs of Theorems~\ref{thm:hovering} and~\ref{thm:amortized} assume a fully random hash function, but it turns out this assumption is not needed. In particular, one can instead use tabulation hashing, and modify the proofs in the preceding sections as follows: analogues of Lemmas \ref{lem:saturated} and \ref{lem:bolus} follow directly from Theorem 8 of \cite{PatrascuTh12}, and then all of the Chernoff bounds throughout the paper can be re-created using Theorem~1 of \cite{PatrascuTh12}. Note that each application of Theorem~8 and Theorem~1 of \cite{PatrascuTh12} introduces a $1/\poly(n)$ failure probability, but this is easily absorbed into the analysis.
\end{remark}

\section{Graveyard Hashing}\label{sec:graveyard}

In this section, we describe and analyze a new variant of linear probing, which we call \defn{graveyard hashing}.
Graveyard hashing takes advantage of the key insight in this paper, which is that tombstones have the ability to significantly improve
insertion performance. 

\paragraph{Description of graveyard hashing.}
Graveyard hashing uses different rebuild window sizes, depending on the load factor. If a rebuild is performed at a load factor of $1 - 1/x$,
then the next rebuild will be performed $n / (4x)$ operations later.\footnote{Note that graveyard hashing counts both insertions and deletions as part of
 the length of a rebuild window.}

Whenever the hash table is rebuilt, Graveyard hashing first removes all of the tombstones that are currently present. It then spreads
\emph{new tombstones} uniformly across the table. If the current load factor is $ 1-1/x $, then $ n/(2x) $
tombstones are created, with one tombstone assigned to each of the hashes $\{2ix \mid  i \in [n / (2x)]\}$.
The purpose of these tombstones is to help all of the up to $n / (4x)$ insertions that occur between the current rebuild
 and the next rebuild.

The insertion of tombstones during rebuilds is the \emph{only difference} between graveyard hashing and standard ordered linear probing.
Thus insertions, queries, and deletions are implemented exactly as in a traditional ordered linear probing hash table.

If desired, one can implement graveyard hashing so that each rebuild also dynamically resizes the table, ensuring that the load factor is always
$1 - \Theta(1 / x)$ for some fixed parameter $x$. Note that, when resizing the table, the elements of the table will need to be assigned to new hashes, and thus will need to be permuted. In the RAM model, this can easily be done in linear time (and in place) using an in-place radix sort. In the external-memory model, resizing can be implemented in $O(n / B)$ block transfers (where $B$ is the external memory block size) using Larson's 
block-transfer efficient scheme for performing partial expansions/contractions on a hash table \cite{Larson82} (this technique has also been used in past work on external-memory hashing, see \cite{JensenPa08, PaghWeYi14}).

\paragraph{Analysis of graveyard hashing.}
To perform the analysis, we will need one last balls-and-bins lemma:

\begin{lemma}
Suppose that $\mu n$ balls are placed into $n$ bins at random.
Let $x > 1$, $ k \ge 1 $, and $j \in [n]$. With probability $ 1-2 ^ {-\Omega (k) } $, for every interval $I \subseteq [n]$ that contains $j$,
the number of balls in the bins $I$ is at most $ (1 + 1/x) |I|\mu + xk$.
\label{lem:saturated3}
\end{lemma}
\begin{proof}
Suppose there is some interval $ I $ satisfying $i \in I$ such that the number of balls in the interval $ I $
is greater than $(1 + 1/x) |I|\mu  + xk$. 
 If $I = [j_0, j_1] $ for some $j_0, j_1$, then we can break $ I $ into two sub-intervals $I_1 = [j_0, j] $ and $I_2 = (j, j_1].$ Since $(1 + 1/x) |I|\mu  + xk$, at least one of the two subintervals $I_k \in \{I_1, I_2\}$ must contain at least 
$$(1 + 1/x)|I_k|\mu + xk/2$$
balls. However, by Lemma \ref{lem:bolus}, the probability of any such subinterval $I_k$ existing is at most $2^{-\Omega(k)}$. 
\end{proof}
\begin{corollary}
Suppose that $\mu n$ balls are placed into $n$ bins at random.
Let $x > 1$, $ k \ge 1 $, and $j \in [n]$. With probability $ 1-2 ^ {-\Omega (k) } $, for every interval $I \subseteq [n]$ that contains $j$ and has size $ |I| \ge x^{2 } k $,
the number of balls in the bins $I$ is at most $ (1 + 1/x) |I|\mu $.
\label{cor:saturated3}
\end{corollary}
\begin{proof}
This follows by applying Lemma \ref{lem:saturated3} with $x' = x / 2$ and $k' = kx / 2$. 
\end{proof}

We now turn our attention to analyzing graveyard hashing. As in the previous sections, it will be easier to analyze the displacement of an insertion rather than directly analyzing the running time of each insertion. 
Recall that the displacement $d_u$ of an insertion $ u $ is $i - h(u)$ where $i$ is the hash of the tombstone that $u$ uses, if $ u $ uses
a tombstone, and $i$ the position of the free slot that $ u $ uses, if $ u $ uses a free slot.

The next lemma bounds the difference between displacement and running time. The fact that the rebuild windows for graveyard hashing are so small
(only $n / (4x)$ operations) ends up allowing for an especially simple argument.

\begin{lemma}
Consider the insertion of an element $ u $. Let $ d$ denote the displacement of the insertion, and $ t $ denote the running time.
Then, for any $r \ge 1$,
$$\Pr[t - d - 1 \ge r x] \le \exp(-\Omega(r)).$$
\label{lem:displacementerror}
\end{lemma}
\begin{proof}
We can assume without loss of generality that $ u $ makes use of some tombstone $v$ (rather than a free slot), since otherwise the lemma is trivial. 
The displacement of $u$ is therefore given by $d = h(v) - h(u)$ and the running time of $u$ is given by $t = k - h(u) + 1$, where $ k $
is the position in which $ v $ resides. Thus
$$\Pr[t - d - 1 \ge r x] = \Pr[k - h(v) \ge r x],$$
which means that we want to show that
$$\Pr[k - h(v) \ge r x]\le \exp(-\Omega(r)).$$

For each element/tombstone $v'$ in the run containing position $h(u)$, define the \defn{placement-error} $e_{v'}$ 
to be $k' - h(v')$, where $ k' $ is the position in which $ v' $ resides (at the moment of time prior to the insertion $ u $).

We wish to show that $ e_{ v } < rx$, but we must be careful about the fact that $ v $ is partially a function of the randomness
of the hash table. In order to bound $e_v$, we assume that $ v $ is selected adversarially, and instead bound the quantity
$$\beta = \max\{ e_{v'} \mid v'\text { is in the same run as } u,\text { and } h (v')\ge h (u)\}. $$
We cannot afford to simply union bound over the different options for $ v' $ here; instead we must make use of the fact that
the values of $e_{v'}$ are closely correlated for different keys $v'$ in the same run.

Let $v'$ be the element for which $\beta = e_{v'}$. Let $ p $ be the position of the left-most element in $u$'s run.
All of the elements/tombstones that reside in positions $p, \ldots, h(v') + e_{v'}$ must have hashes in $[p, h(v')]$. 
The number of elements/tombstones (at the time prior to $u$'s insertion) that hash to the interval $I = [p, h(v')]$ must therefore
be at least $|I| + e_{v'}$.
In contrast, the expected number of elements/tombstones that hash into the interval $I$ is that most $|I|$. 
Thus there are at least $e_{v'}$ more element/tombstones that hash into $I$ then are expected.
By Lemma \ref{lem:saturated3}, it follows that $\Pr[e_{v'} \ge r x] \le \exp(-\Omega(r))$.
\end{proof}

Graveyard hashing is designed so that there are always copious tombstones for insertions to make use of. Note, in particular, that each rebuild window begins with
$n / (2x)$ tombstones but only contains at most $n / (4x)$ insertions. This allows for the following bound on displacement.
\begin{lemma}
Consider an insertion $ u$. The displacement $ d $ of $ u $ satisfies
$$\Pr [d \ge r x] \le \exp(-\Omega(r))$$
for all $ r >1 $.
\label{lem:graveyardbound}
\end{lemma}
\begin{proof}
Call a tombstone \defn{primitive} if it was inserted during the rebuild prior to the current rebuild window.
Let $ T $ be the set of primitive tombstones present when $ u $ is inserted. Let
$$ j_0 =\max\{ h (v)\mid v\in T, h (v) < h (u)\}$$
and 
$$ j_1 =\min\{ h (v)\mid v\in T, h (v)\ge h (u)\}.$$ 
The displacement $ d $ is at most $ j_1 - h (u) \le j_1 - j_0$.
To complete the proof, we will bound the probability that 
$ j_1  - j_0 \ge rx $. 

At the beginning of the rebuild window, there were 
$\frac{j_1 - j_0}{2x} - 1 $ primitive tombstones with hashes in the range
$ I = (j_0, j_1)$; denote the set of these tombstones by $L$. By the time $ u $ is inserted,
all of the tombstones $L$ have been used by insertions (this is by the definition of $j_0$ and $j_1$). Since there is still a primitive tombstone 
with hash $ j_0 $, the insertions that used up $ L $ must have all had hashes at least $ j_0 + 1 $.
Thus, during the current rebuild window, there have been at least $|L| = \frac{j_1 - j_0}{2x} - 1 $ insertions that hashed into
the interval $I = (j_0, j_1)$. 

Recall, however, that each rebuild window consists of only $n / (4x)$ operations. The expected number of insertions that hash into $I$
 is therefore a most $\frac{j_1 - j_0}{4x}$. 

In summary, the only way have $j_1 - j_0 \ge rx$ is if (1) there is an interval $I$ containing $h(u)$ that satisfies $|I| \ge rx - 2$,
and (2) the number $q$ of insertions (during the current rebuild window) that hash into $ I $ is a constant factor larger than the expected
number $\E[q]$ of such insertions. To bound the probability of such an $ I $ existing, we partition the slots of the hash table into ``bins'' of size $x$,
and treat keys inserted during the rebuild window as balls that each hash to a bin. In order for $ I $ to exist, there must be a contiguous subsequence of $\Theta(r)$ bins such that the interval of hashes covered by the bins contains $h(u)$, and such that the bins contain a constant factor more balls than expected. By Corollary \ref{cor:saturated3}, the probability 
of such a subsequence of bins existing is that most $\exp(-\Omega(r))$.
\end{proof} 

Combining the previous lemmas, we can analyze the running time of graveyard hashing.
\begin{theorem}
Consider a graveyard hash table. For each insertion/query/deletion, if the 
operation is performed at some load factor of $1 - 1/x$ then the operation takes
expected time $O(x)$, and incurs amortized time $O(x)$ for rebuilds.
\label{thm:graveyard}
\end{theorem}
\begin{proof}
Since graveyard hashing uses small rebuild windows (i.e., of size $n / (4x)$), we can analyze queries by ignoring the
deletions in the rebuild window, and applying the classic $O(x)$ bound for query time in 
an insertion-only table (Proposition \ref{prop:classic2}). Deletions of keys $ u $ take the same amount of time that a query for that key would have, so deletions also take expected time $ O (x) $. To analyze insertions, we can apply Lemmas \ref{lem:displacementerror} and \ref{lem:graveyardbound}, which together bound the
expected time by $O(x)$.

Finally, we analyze the amortized cost of rebuilds per operation. If a rebuild window starts at a load factor of $1 - 1/x$, then the next rebuild is performed after $\Theta(n/x)$ operations, and all of those operations are performed at load factors $1 - \Theta(1/x)$. The rebuild
can be performed in time $\Theta(n)$ and thus the amortized cost per operation is $\Theta(x)$.
\end{proof}

\begin{remark}
The proof of Theorem \ref{thm:graveyard} assumes a fully random hash function, but this assumption is not necessary. The theorem continues to hold if we use either tabulation hashing or 5-independent hashing. In particular, one can use Equation (23) from \cite{PatrascuTh12} as a substitute for Lemma \ref{lem:saturated3}, and then re-create all of the proofs above without modification. 
\end{remark}

We conclude the section by analyzing graveyard hashing in the external-memory model \cite{AggarwalVi88}.
\begin{theorem}
Consider the external memory model with $B = rx$ for some $r \ge 1$ and $M = \Omega(B)$. 
Graveyard hashing can be implemented to offer the following guarantee on any sequence of operations. The load factor of 
the table is maintained to be $1 - 1 / \Theta(x)$ at all times, and each operation incurs
$$1 + O(1 / r)$$
block transfers in expectation. Furthermore, the amortized block-transfer cost (per operation) of rebuilds is $ O(1 / r)$. 
\end{theorem}
\begin{proof}
By Theorem~\ref{thm:graveyard}, the expected time taken by a given operation is $O(x)$. It follows that the expected number of block boundaries that are crossed by the operation is $O(x / B) = O(1/r)$. Thus the expected number of block transfers incurred is $O(1/r)$.

Next we analyze the cost of rebuilds. Each rebuild window contains $\Theta(n / x)$ operations at a load factor of $1 - \Theta(1 / x)$, and, as discussed earlier, the rebuild at the end of the window can be implemented with $O(n / B)$ block transfers. This implies the desired bound of $ O(1 / r)$ on amortized rebuild cost.
\end{proof}

\begin{corollary}
If $x = o(B)$, then the amortized expected number of block transfers per operation is $1 + o(1)$.
\end{corollary}



\section{Related Work}
\label{sec:related}

\paragraph{Alternative probing methods.}
Beginning in the late 1960s, there was a flurry of work on alteratives to linear probing.  
A central question has been whether one can have probe sequences that benefit from the data locality of traditional linear probing while also eliminating primary clustering.

At one extreme is uniform probing, where each probe is to a random location, thus sacrificing locality in the probe sequence~\cite{Peterson57}. There has been intensive work in analyzing variations on uniform probing, including in the presence of deletions~\cite{CelisFr92,JimenezMa18,Mitzenmacher16-robinhood,Mitzenmacher16-doublehashing,VanWykVi86,KenyonVi91,Larson83,Ramakrishna89}.

Double hashing~\cite{Knuth98Vol3,Brent73,GuibasSz78,WikipediaDoubleHashing,Mitzenmacher16-doublehashing,LuekerMo88,LuekerMo93} is a classic alternative to uniform probing, in which a primary hash function determines the first probe and a secondary hash function determines jump size between indices in the probe sequence. Double hashing has been shown to have short probe sequences similar to that of uniform hashing, but like uniform probing, these short probe sequences come with a corresponding loss in locality.

In 1968, Maurer~\cite{Maurer68,HopgoodDa72} introduced quadratic probing,
which remains a widely used solution today (see, e.g., \cite{Abseil17, BronsonSh19}). Although Maurer's original scheme used a probing sequence
that cycled after $n / 2$ iterations, subsequent work~\cite{HopgoodDa72}
has shown how to construct quadratic probing sequences that hit every position in the table. Quadratic probing can be viewed as a hybrid of linear probing and double hashing, maintaining some of the spatial locality of the former, while empirically obtaining probe complexity similar to the latter \cite{Maurer68}. 



\paragraph{Cuckoo hashing.}
In addition to probing and chaining, another form of hash table that has become widely used (see, e.g.,~\cite{LiAnKa14,SunHuJi17,PontarelliRe16}) is cuckoo hashing, 
which was introduced in 2004 by Pagh and Rodler~\cite{PaghRo01, PaghRo04}. Cuckoo hashing guarantees that every record
$u$ is in one of two positions $h_1(u)$ or $h_2(u)$ in the hash table. The result is that, even in the worst case,
queries take time $O(1)$. Although originally cuckoo hashing
only supported load factors smaller than $ 1/2 $, subsequent work has shown how to support higher load factors by
either (1)~using $d > 2$ hash functions \cite{FotakisPaSa05}, or (2)~hashing records $u$ to bins $h_1(u)$ and $h_2(u)$ (rather than to slots)
that each have some capacity $c > 1$ \cite{DietzfelbingerWe07}.

A drawback of cuckoo hashing is that it is less I/O efficient than linear probing. Negative queries always require
$d \ge 2$ I/Os, and insertions at high load factors require $\omega(1)$ I/Os~\cite{FotakisPaSa05,DietzfelbingerWe07}.
In contrast, we show that linear probing can offer $1 + o(1)$ I/Os per operation, as long as the memory block size is $\omega(x)$. 

\paragraph{Other work on hashing.}
Broadly speaking, work on hashing can be categorized into three major categories. The first, discussed above, has been to understand and try to improve the behavior of three core hash-table designs: probing, cuckoo hashing, and, to a lesser extent, chaining. The second direction has been to study what types of additional features are possible to achieve in a hash table (or, more generally, a dictionary). And the third direction has been to construct explicit families of hash functions that can be used in place of full independence.

As an example of a major result in the second category, Dietzfelbinger and Meyer auf der Heide~\cite{DietzfelbingerHe90} showed that it is possible to achieve
\emph{worst case} constant time operations (with high probability), building on prior work by Fredman et al.~\cite{FredmanKoSz82} 
and by Dietzfelbinger et al.~\cite{DietzfelbingerKaMe88}. Subsequent work has pushed this guarantee even further, showing that it is also possible to have a $1 - o(1)$ load factor \cite{DemaineMePa06, ArbitmanNaSe09, LiuYiYu20},
as well as a sub-polynomial failure probability~\cite{GoodrichHiMi12, GoodrichHiMi11}.
There has also been a series of work on upper and lower bounds for deterministic dictionaries~\cite{Sundar91,HagerupMiPa01, Ruzic08,Pagh00,PatrascuTh14} and external-memory hashing~\cite{JensenPa08,IaconoPa12,ConwayFaSh18,PaghWeYi14,PaghWeYi10},
as well as work on achieving security guarantees such as history independence~\cite{BlellochGo07,NaorSeWi08,NaorTe01,BenderBeJo16} and protection against an adversary that can see where in memory is being accessed~\cite{ChanGuLi17,YuqunVeCa02}.

In the third category, that is, the study of explicit hash function families, 
there are now some quite universal results known, including families of hash functions \cite{naor1999construction, luby1988construct, pagh2008uniform, KaplanNaRe09, DietzfelbingerWo03} that can be made compatible with essentially any hash table (see, for example, the usage in \cite{ArbitmanNaSe10, LiuYiYu20}). 

Additionally, there has been quite a bit of work on families of hash functions for specific hash tables, especially linear probing \cite{PaghPaRu07, MitzenmacherVa08, PatrascuTh10, PatrascuTh12, PaghPaRu11, Thorup09} and cuckoo hashing \cite{DietzfelbingerSc09, AumullerDiWo16, AumullerDiWo14, PatrascuTh12, Thorup17}. Linear probing, in particular, has been shown to be compatible with several especially simple families: Pagh et al. \cite{PaghPaRu07, PaghPaRu11}  showed that \emph{any} 5-independent family suffices, and P\v{a}tra\c{s}cu And Thorup \cite{PatrascuTh12} showed that the family of simple tabulation hash functions also suffices.
(In Section~\ref{sec:graveyard}, we describe how these results can be applied to graveyard hashing as well.) Although linear probing is, in general, \emph{not} compatible
with all 4-independent families of hash functions \cite{PatrascuTh10}, one of the most surprising results in the area is that of Mitzenmacher and Vadhan \cite{MitzenmacherVa08}, which establishes that, in any workload with sufficient entropy, even 2-independence suffices.

\paragraph{Relationship with filters.}
One of the most widely used applications of hash tables at high load factors is for the construction of \defn{filters}, 
which are compact dictionaries that supports some false-positive rate $\epsilon$.
The classic filter is the Bloom filter \cite{Bloom70}, which has inspired numerous variants~\cite{QiaoLiCh14, PutzeSaSi07,
LuDeDu11, CanimMiBh10, DebnathSeLi11, AlmeidaBaPr07, FanCaAl00, BonomiMiPa06,BroderMi03}.
Whereas the Bloom filter supports only a limited set of operations (no deletions) and no resizing, 
modern filters have shown how to adapt space-efficient hash tables 
in order to construct practical space-efficient filters that support a richer set of operations. Quotient filters and variants (e.g., counting and vector)~\cite{BenderFaJo12,PandeyBeJo17,BenderFaJo12,EinzigerFr16,GeilFO18,PandeyCoDu21} are built on the idea of storing small
fingerprints via ordered linear probing~\cite{AmbleKn74}; and cuckoo~\cite{FanAnKa14} and Morton filters~\cite{BreslowJ18} are based on the idea of storing these fingerprints via cuckoo hashing~\cite{PaghRo01, PaghRo04}.
There has also been an effort to push forward the theoretical frontiers of
what guarantees a filter can offer (see, e.g.,  Pagh's  
single-hash filter~\cite{PaghRaRa05} as well as more recent results~\cite{BenderFaGo18, LiuYiYu20}).

The fact that many filters are implemented on top of hash-table designs means that improvements to hash table performance
directly results in improvements to filter performance. For example, our techniques for improving linear probing can immediately be applied to linear-probing based filters~\cite{BenderFaJo12,PandeyBeJo17,EinzigerFr16,GeilFO18,PandeyCoDu21}.

\paragraph{Relationship with other data structures.}
One of the interesting features of our results in that it reveals an unexpected connection between the linear probing and several other problems in data structures (e.g., list labeling \cite{BulanekKoSa12}, file maintenance \cite{ItaiKoRo81,Willard86,Willard82,Willard92,BenderFiGi17,ItaiKa07}, cache-oblivious and locality-preserving B-trees \cite{BenderDeFa00,BenderDuIa04,BrodalFaJa02}, and even sorting \cite{BenderFaMo06}). Solutions to these problems all share a commonality, which is that they strategically leave extra space between elements of a data structure in order to enable fast modifications. 
One  interesting aspect of linear probing is that this ``extra space'' (tombstones and free slots) \emph{already} naturally occurs spread throughout the table, but that different forms of extra space (i.e., tombstones versus free slots) end up interacting very differently with operations of the hash table.

\markeverypar{\the\everypar\looseness=0 }

\newcommand{\bibcacm}{Communications of the ACM}
\newcommand{\bibinfsci}{Information Sciences}
\newcommand{\bibtheoreticalcs}{Theoretical Computer Science}

\bibliographystyle{plainurl}

\bibliography{hashing,bibliography,pma}

\appendix
\section{Linear Probing in Textbooks and Courses}\label{sec:books}

\begin{center}
\begin{tabular}{|l|c|c|l|}
\multicolumn{4}{c}{Textbooks' Stances on Linear Probing}\\
\hline
Source                                    & Teaches     & Teaches  & Recommends \\
                                          & primary     & Knuth's  & \multicolumn{1}{r|}{(QP=quadratic probing)}\\
                                          & clustering  & formulae & \multicolumn{1}{r|}{(DH=double hashing)}\\
\hline
Cormen, Leiserson, Rivest, and Stein \cite{CormenLeRi09}    & yes & no  & QP or DH \\
Dasgupta,  Papadimitriou, and Vazirani\cite{DasguptaPaVa06} & no & no & not applicable \\
Drozdek and Simon \cite{DrozdekSi95}      & yes & yes & QP or DH \\
Goodrich and Tamassia \cite{GoodrichTa15} & yes & no  & $\le 50\%$ load factor \\
Kleinberg and Tardos \cite{KleinbergTa06} & no  & no  & not applicable \\
Kruse \cite{Kruse84}                      & yes & yes & chaining \\
Lewis and Denenberg \cite{LewisDe91}      & yes & yes & DH with ordered probing \\
Main and Savitch \cite{MainSa01}          & yes & yes & DH \\
McMillan \cite{McMillan14}                & yes & no  & chaining \\
Sedgewick \cite{Sedgewick83,Sedgewick90}  & yes & yes & chaining or DH \\
Standish \cite{Standish95}                & yes & yes & not prescriptive \\
Tremblay and Sorenson \cite{TremblaySo84} &     & yes & DH \\
Weiss \cite{Weiss00}                      & yes & yes & QP \\
Wengrow \cite{Wengrow17}                  & no  & no  & $\le 70\%$ load factor \\
Wikipedia \cite{WikipediaLinearProbing}   & yes & yes & QP or DH \\
Wirth \cite{Wirth76,Wirth86}              & yes & partly & search trees \\

Wirth \cite{Wirth76,Wirth86}              & yes & partly & search trees \\
\hline
\multicolumn{4}{c}{}\\
\multicolumn{4}{c}{Some Course Notes' Stances on Linear Probing}\\
\hline
Source                     & Teaches     & Teaches  & Recommends \\
                                          & primary     & Knuth's  & \\
                                          & clustering  & formulae & \\
\hline
CMU Systems \cite{Kesden07}               & yes & no & QP or DH \\
CalPoly Fundamentals of CS \cite{Fisher01} & yes & no & QP or DH \\
Columbia Data Structs in Java \cite{Bauer15} & yes & no & QP or DH \\
Cornell Prog.~and Data Structs \cite{GriesJa14} & partly & partly & QP \\
Harvard Intro.~to CS \cite{Sullivan21} & yes & no & QP or DH \\
MIT Advanced Data Structs \cite{Demaine12-6897-lecture10} & yes & no & low load factor \\
MIT Intro. to Algorithms \cite{DemaineLe05} & yes & no & DH \\
Stanford Data Structs \cite{Schwarz21}  & yes & yes & chaining or low load factor \\
UIUC Algorithms \cite{Erickson17} & yes & no & binary probing \\
UMD Data Structs \cite{Mount19} & yes & yes & DH \\
UT Software Design \cite{Mitra21}    & yes & no & $\le 2/3$ load factor \\
UW Data Structs and Algorithms \cite{deGreef17} & yes & yes & QP or DH \\
\hline
\end{tabular}
\end{center}




\newcommand{\punt}[1]{}
\punt{
\paragraph{Punt this: Some citations (in chronological order)}
\begin{itemize}
  \item Wirth \cite{Wirth76} describes chaining ("has the disadvantage that secondary lists must be maintained and that each entry must provide space for a pointer"), linear probing ("has the disadvantage that entries have a tendency to cluster around the primary keys"), and double hashing ("too costly"), and quadratic probing "a comprise by being simple to compute and still superior to [linear probing]", "essentially avoids primary clustering", "slight disadvantage is that not all table entries are searched" (with a proof that quadratic numbers only hit half of $Z/p$). Fails to distinguish between the cost of a successful search and an unsuccessful search, and so says that linear probing costs $(1-\alpha/2)/(1-\alpha)$ probes (which looks right for successful searches).  He thinks linear probing, "the poorest method of collision handling" is so good that there is "a temptation to regard ... hashing .. as the panacea".  He then proceeds to recommend search trees because "the size of the [hash] table is fixed and cannot be adjusted" and that "deletion is extremely cumbersome unless direct chaining ... is used", and because you need to sort the output anyway, so hashing hasn't really saved you anything.)  Wirth hates arrays anyway: he thinks that the length of the array is part of its type (e.g., in his merge sort, the size of the array is a compile-time constant.)  Wirth's confusion persisted unchanged in his 1986 textbook \cite{Wirth86}.
  
  \item Kruse \cite{Kruse84} explains linear probing (``as the table becomes about half full, there is a tendency toward \textit{clustering}; that is, records start to appear in long strings of adjacent positions with gaps between the strings.  Thus the sequential searches needed to find an empty position become longer and longer.'') and quadratic probing."  "The problem of clustering is thus one of instability; if a few keys happen randomly to be near each other, then it becomes more and more likely that other keys will join them, and the distribution will become progressively more unbalanced" (this analysis seems to overemphasize winner-keeps winning".   Explains quadratic probing, double hashing.  Recommends not to delete from hash tables because "deletions are indeed awkward".  
  
  \item Lewis and Denenberg \cite{LewisDe91} points out that primary clustering is caused both by winner-keeps-winning and by globbing small runs together into big runs.  Gives Knuth's formulae for successful and successful search.   (And concludes that double hashing is a big win.)  Points out that double hashing with ordered probing gives an history-independent data structure, and concludes that successful and unsuccessful searches both cost $(1/\alpha) ln 1/(1-\alpha)$   Describes tombstones and complains that they slow down operations (but doesn't analyze it) and recommends rebuilding if there are too many.  No mention of quadratic probing.

  \item McMillan \cite{McMillan14} recommends linear probing when the load factor of the table is less than $1/2$, and otherwise recommends chaining.
  
  \item Weiss \cite{Weiss00} says "linear probing is not a terrible strategy" (for the case where the load factor is 0.5), pointing out that something else can only save half a probe on average.  Then proceeds to recommend quadratic probing (which he says requires that the load $\le 0.5$.  "Quadratic probing has not yet been mathematically analyzed, although we know that it eliminates primary clustering".  Doesn't explain it as ``winner keeps winning'', but punts on an explanation.
  
  \item Wengrow \cite{Wengrow17} describes a form of chaining in which each linked list is instead a dynamically-sized array, and recommends a load factor of 0.7.
  
 \item \cite{MainSa01} doesn't make the claim that it's ``winner keeps winning''.  They understand it as ``the clusters tend to merge into larger and larger clusters''.

 \item \cite{Standish95} mentions clusters merging as well as winner keeps winning as being the causes for primary clustering.

\end{itemize}

Secondary clustering, which is a low-order-term is discussed by \cite{TremblaySo84}.

Secondary clustering caused by a bad hash function: \cite{WikipediaPrimaryClustering}
}

\end{document}